\documentclass[format=acmsmall, review=false]{acmart}

\usepackage{acm-ec-20} 

\usepackage{amsthm}

\usepackage{algorithm}
\usepackage[noend]{algpseudocode}

\usepackage{pgfplots}

\usepackage{mathtools}
 \usepackage{graphicx}
 \usepackage{float}
 \usepackage{tikz, xcolor}
 \usetikzlibrary{arrows.meta}
\usetikzlibrary {shapes}
\usepackage{footmisc}

\usepackage[most]{tcolorbox}

  \usetikzlibrary{decorations.pathmorphing}
  
  \usetikzlibrary{decorations.pathreplacing}
  \usetikzlibrary{positioning, calc}

\usepackage{tcolorbox}

\setcitestyle{acmnumeric}

\newcommand*{\draft}{}

\ifdefined\draft
    \newcommand{\patrick}[1]{\textcolor{blue}{[Patrick: #1]}}
    \newcommand{\andy}[1]{\textcolor{red}{[Andy: #1]}}
    \newcommand{\brendan}[1]{\textcolor{orange}{[Brendan: #1]}}
\else
    \newcommand{\patrick}[1]{}
    \newcommand{\andy}[1]{}
    \newcommand{\brendan}[1]{}
\fi

\title{The Carnot Bound: Limits and Possibilities for Bandwidth-Efficient Consensus}

\date{Nov 2025}

\author{Andrew Lewis-Pye\textsuperscript{1,2}}
\author{Patrick O'Grady\textsuperscript{2}}

\renewcommand{\shortauthors}{Lewis-Pye and O'Grady}

\affiliation{%
\textsuperscript{1}\institution{London School of Economics}
\country{UK}
}

\affiliation{%
\textsuperscript{2}\institution{Commonware}
\country{USA}
}

\begin{abstract}
In leader-based protocols for State Machine Replication (SMR), the 
leader's outgoing bandwidth is a natural throughput bottleneck. 
Erasure coding can alleviate this by allowing the leader to send 
each processor a single fragment of each block, rather than a full 
copy. The \emph{data expansion rate}, the ratio of total data sent 
to payload size, determines how close throughput can get to the 
underlying network bandwidth.

We investigate the fundamental limits and possibilities for 
bandwidth-efficient leader-based consensus. On the negative side, we 
prove that protocols with 2-round finality (one round of voting) 
cannot achieve a data expansion rate below approximately $2.5$, a 
bound that is matched by existing protocols. On the positive side, 
we show that protocols with 3-round finality (two rounds of voting) 
can do significantly better. The key insight is that the second 
voting round provides a recovery mechanism: leaders can attempt 
aggressive erasure codes and safely fall back to more conservative 
ones when reconstruction fails, without compromising consistency.

We present two protocols with 3-round finality realising this 
approach, targeting two natural formulations of the replication 
problem. Carnot~1 solves \emph{Extractable SMR}, in which any 
correct processor can efficiently reconstruct any finalised block 
from the fragments held by correct processors, but processors are 
not required to hold full blocks locally. Extractable SMR suffices 
for settings such as data availability layers, where on-demand 
reconstruction is enough. Carnot~1 assumes $n \geq 4f+1$ 
processors (of which at most $f$ may be Byzantine) and achieves a 
clean design requiring no additional fragment dissemination beyond 
the initial protocol messages. Carnot~2 solves full \emph{SMR}, in which 
every correct processor eventually receives every finalised 
transaction. It operates under the optimal resilience assumption 
$n \geq 3f+1$, at the cost of additional fragment dissemination
when Byzantine processors interfere. Both protocols  can incorporate
multi-view leaders  to maximise throughput and
minimise latency: since a multi-view leader pipelines proposals,
block times are limited only by bandwidth and per-view protocol
overhead, and are of the order of a millisecond in realistic
settings. Under favourable conditions, both protocols allow
leaders to use data expansion rates approaching~$1$; under 
adversarial conditions, leaders can revert to safe expansion rates 
of approximately $1.33$ and $1.5$, respectively---both still well 
below the $2.5$ lower bound for protocols with 2-round finality.
\end{abstract}

\begin{document}

\maketitle

\section{Introduction} \label{intro}

State Machine Replication (SMR)~\cite{schneider1990implementing} is a fundamental primitive for distributed computing that allows a collection of processors to maintain a consistent, shared log of transactions despite the failure, or even malicious behaviour, of some participants. Originally developed for fault-tolerant systems, SMR has become the algorithmic backbone of modern blockchains and decentralised applications. The core challenge is to ensure that all correct processors agree on the same sequence of transactions (Consistency), while also guaranteeing that new transactions are eventually included (Liveness), even in the presence of \emph{Byzantine} faults, where corrupted processors may behave arbitrarily.

Protocols for SMR typically operate in the \emph{partially synchronous} model~\cite{DLS88}, in which messages may be delayed arbitrarily before an unknown time GST (the Global Stabilisation Time), but are delivered within a known bound $\Delta$ thereafter. This model captures realistic network conditions, including temporary partitions, and underpins many widely deployed protocols, from PBFT~\cite{castro1999practical} to the consensus layers of modern blockchains such as Ethereum~\cite{buterin2017casper}.

\vspace{0.2cm}
\noindent \textbf{The throughput challenge.} A central goal in recent protocol design is to maximise \emph{throughput}: the rate at which transactions can be finalised. Transactions are typically grouped into \emph{blocks}, and protocols proceed by having processors agree on a sequence of blocks. Two broad families of protocols have emerged. In \emph{DAG-based} protocols~\cite{keidar2021all,spiegelman2022bullshark,shrestha2025sailfish,keidar2022cordial}, multiple processors propose blocks concurrently, organising them into a directed acyclic graph from which a linear ordering is then extracted. In \emph{leader-based} protocols~\cite{castro1999practical,yin2019hotstuff,shoup2025kudzu,chou2025minimmit}, a single designated leader proposes each block, and the remaining processors vote to accept or reject it.

DAG-based protocols can achieve high throughput through their concurrency, but at the cost of higher \emph{round-latency},\footnote{State-of-the-art DAG-based protocols, such as Sailfish~\cite{shrestha2025sailfish}, can finalise transactions in designated `leader blocks' in the same optimal number of communication rounds as leader-based protocols, but transactions in most blocks require additional rounds before finalisation.} in the sense that finalisation requires more rounds of communication. Leader-based protocols offer lower round-latency, but the leader's outgoing bandwidth is a natural bottleneck, which can be alleviated using \emph{erasure coding}. Rather than sending a full copy of each block to every processor, the leader encodes the block's payload and sends each processor a single fragment, from which the full payload can be reconstructed once sufficiently many fragments are collected. The efficiency of this encoding is captured by the \emph{data expansion rate}: if the payload has size $\beta$, then the total data the leader must send is approximately $d \cdot \beta$, where $d$ is the data expansion rate. As examined by Lewis-Pye, Nayak and Shrestha~\cite{lewis2025pipes}, this factor $d$ directly governs the throughput bottleneck. For a network in which each processor has bandwidth $S$, i.e., processors  can send $S$ bits per unit time, payloads of size $\beta$ require approximately $d\beta / S$ time to disseminate. The maximum throughput, measured in payload bits finalised per unit time, is therefore approximately $S/d$. To approach throughputs matching the underlying network bandwidth, one must push the data expansion rate as close to $1$ as possible.

\vspace{0.2cm}
\noindent \textbf{Finality and voting rounds.} Before stating our results, we establish some basic terminology. We consider a set of $n$ processors, of which at most $f$ may display Byzantine (arbitrary) behaviour. In leader-based protocols, which are the focus of this paper, the protocol proceeds in sequential \emph{views}, each with a designated \emph{leader} who assembles transactions into a block and proposes it to the other processors, who then vote on whether to accept it. We use the term \emph{$r$-round finality} to describe a protocol in which finality requires $r$ rounds of communication: one round for the leader's proposal, followed by $r-1$ rounds of voting. Protocols with \emph{2-round finality}, consisting of one proposal round and one voting round, include Minimmit~\cite{chou2025minimmit}, Hydrangea~\cite{shrestha2025hydrangea}, and Kudzu~\cite{shoup2025kudzu}, and typically assume $n \geq 5f+1$ ($n\geq 5f-1$ is necessary and sufficient~\cite{kuznetsov2021revisiting}). Protocols with \emph{3-round finality}, consisting of one proposal round and two voting rounds, include Simplex~\cite{chan2023simplex} and the protocols we present in this paper, and typically operate under the (optimal) assumption that $n\geq 3f+1$.

\vspace{0.2cm}
\noindent \textbf{Our contributions.} This paper investigates the fundamental limits and possibilities for bandwidth-efficient leader-based consensus using erasure coding. We establish a lower bound on the data expansion rate for protocols with 2-round finality, and show that protocols with 3-round finality can do significantly better.

\vspace{0.1cm}
\noindent \emph{An impossibility result for 2-round finality (Section~\ref{imposs}).} We prove that protocols with 2-round finality cannot achieve a data expansion rate below $2.5$. This bound is tight: E-Minimmit~\cite{chou2025minimmit} and Kudzu~\cite{shoup2025kudzu} achieve data expansion rates of approximately $2.5$.

\vspace{0.1cm}
\noindent \emph{Circumventing the bound with 3-round finality.} Protocols with 3-round finality can do significantly better. 
Indeed, some versions of DispersedSimplex~\cite{shoup2023sing} already achieves a 
data expansion rate of approximately $1.5$ with 3-round finality 
(using an erasure coding approach similar to that described in Section~\ref{1.5} of this paper). 
The question is whether one can push the rate further toward~$1$. 
The key insight is that the additional voting round in 3-round 
finality protocols provides a crucial \emph{recovery mechanism}.  When the leader uses an aggressive erasure code (with a high reconstruction threshold $k$), it is possible that some correct processors reconstruct the payload, perhaps with the aid of fragments from Byzantine processors, while others cannot. In a 2-round finality protocol, this possibility can lead to unrecoverable configurations, in which consistency and liveness cannot both be maintained. With two rounds of voting, however, the second round can detect whether data availability has been achieved and, if not, allow the protocol to nullify the view and retry. This means leaders can \emph{attempt} low data expansion rates and safely fall back to higher rates when reconstruction fails.

Specifically, a correct leader can use an $(n,k)$-erasure code for any $k$ up to $n-1$,\footnote{We consider $k<n$ so that the leader need not disseminate their own fragment.} meaning that $k$ of $n$ fragments suffice for reconstruction. This produces  a data expansion rate of $n/k$. In the optimistic case, when the leader is correct, the network is synchronous, and at most $f_a$ processors are actually faulty, the leader can set $k$ as high as $n - f_a - 1$, giving a data expansion rate that approaches $1$ as $f_a$ becomes small relative to $n$. When conditions deteriorate, the protocol automatically reverts to a safe expansion rate without interrupting progress.

\vspace{0.1cm}
\noindent \emph{Two formulations of the replication problem.} We present two protocols realising this approach, targeting two natural formulations of the replication problem. In standard SMR, every correct processor eventually receives every finalised transaction. A weaker but practically important variant is \emph{Extractable SMR}~\cite{lewis2025morpheus}, in which correct processors agree on the same sequence of finalised blocks, and any correct processor can efficiently reconstruct any finalised block from the fragments held by correct processors collectively, but processors are not required to hold full blocks locally. Indeed, many well-known protocols, including HotStuff~\cite{yin2019hotstuff}, implicitly solve Extractable SMR rather than full SMR: they ensure data availability for finalised blocks but do not specify the mechanism by which every correct processor retrieves the full payload. Extractable SMR suffices for settings such as data availability layers, where on-demand reconstruction is enough.

\vspace{0.1cm}
The two protocols we present are Carnot\footnote{The name is inspired by the Carnot heat engine, which achieves the theoretical maximum efficiency for converting heat into work. Similarly, our protocols aim to approach the theoretical maximum efficiency for converting network bandwidth into throughput.}~1, a protocol for Extractable SMR assuming $n \geq 4f+1$, and Carnot~2, which solves full SMR under the weaker (and optimal) assumption $n \geq 3f+1$. The safe data expansion rate depends on the resilience assumption: approximately $1.33$ for $n \geq 4f+1$, or approximately $1.5$ for $n \geq 3f+1$.

\vspace{0.1cm}
\noindent \emph{Carnot~1: a clean design solving Extractable SMR for $n \geq 4f+1$.} Our first protocol assumes $n \geq 4f+1$ and achieves a particularly clean design. A key concern when using aggressive erasure codes is ensuring \emph{data availability}: if a correct processor advances to the next view on the basis of a block whose payload it cannot fully reconstruct, it must be guaranteed that the payload can eventually be recovered from the messages received by correct processors collectively. In many protocols, this guarantee requires correct processors to perform additional rounds of fragment dissemination when the leader is Byzantine, adding communication overhead precisely in the cases where resources are most constrained. Carnot~1 avoids this entirely: processors echo their certified fragment once upon receiving it from the leader and voting, but no further echoing is ever required. The protocol achieves this through a careful choice of quorum sizes that guarantees every view produces either a compact data-availability certificate or a mechanism for skipping the view, without additional communication.

\vspace{0.2cm} 
Carnot~1 also incorporates \emph{multi-view leaders}, where a single leader proposes blocks across multiple consecutive views (a \emph{superview}), to eliminate the inter-proposal gap that arises with rotating leaders.

\vspace{0.2cm}
\noindent \emph{Carnot~2: an SMR protocol for $n \geq 3f+1$.} Our second protocol solves full SMR under the standard optimal resilience assumption $n \geq 3f+1$, at the cost of requiring additional fragment dissemination when Byzantine processors interfere. When a processor reconstructs a payload but receives evidence that others have not done so, it re-encodes the payload using a more conservative erasure code and disseminates the resulting fragments, ensuring that all correct processors can eventually reconstruct the payload as well. Under standard operation, with reliable communication and when processors act correctly, no extra dissemination occurs.  Carnot~2 also incorporates multi-view leaders.


\vspace{0.1cm}
Both protocols build on the Simplex consensus protocol~\cite{chan2023simplex}, a simple and efficient protocol for partial synchrony with rotating leaders and two rounds of voting per view. Our protocols extend Simplex with erasure coding, multi-view leaders, and the mechanisms described above. To summarise the trade-off between the two: Carnot~1 is simpler and more communication-efficient under adversarial conditions, but requires a stronger resilience assumption and solves Extractable SMR; Carnot~2 tolerates the optimal number of faults and solves full SMR, but may require additional communication when Byzantine processors interfere.

\vspace{0.2cm}
\noindent \emph{Small block times.} While our main focus is bandwidth efficiency, the use of multi-view leaders has a further consequence that deserves emphasis: \emph{small block times}. A recent line of work~\cite{scaffino2026gatling,elsheimy2026cadence} aims to minimise the \emph{inter-proposal time}, i.e., the time between consecutive block proposals, pushing it below the message delay bound $\Delta$. Since a multi-view leader pipelines proposals---it begins sending the block for each view as soon as it has finished sending the block for the previous view, without waiting for votes or certificates---the inter-proposal time for Carnot is bounded neither by $\Delta$ nor by the actual message delay $\delta$. Instead, block times are determined by bandwidth and per-view protocol overhead alone. The analysis of Section~\ref{anal2} makes this precise: at equilibrium, the time between successive proposals is $(13+\log n)n\lambda/(S-Dd)$, where $\lambda$ is the length of a hash, $S$ is per-processor bandwidth, $D$ is the incoming transaction rate, and $d$ is the data expansion rate. For realistic parameter values this is of the order of a millisecond. Block times of this order also clarify a point of terminology. Leaders that control multiple consecutive views are commonly called `stable' or `slowly rotating', but both terms suggest tenures that are long in real time: with millisecond block times, a leader controlling 100 consecutive views holds office for a fraction of a second, which is less time than a single slot in many deployed protocols. Since it is the wall-clock length of a leader's tenure that matters for concerns such as censorship, a superview is perhaps best seen as taking the single block the leader would propose in a protocol with larger block times, and breaking it into many small blocks, each finalised with low latency. We therefore adopt the (hopefully less misleading) term \emph{multi-view leaders}. We also note that Carnot achieves its small block times without the synchronised clocks required by~\cite{scaffino2026gatling,elsheimy2026cadence}, and without sacrificing \emph{predictable validity}: the guarantee that a proposer can validate the transactions it includes against an up-to-date state of the log at proposal time. Indeed, as discussed in Section~\ref{rw}, the variants of Gatling~\cite{scaffino2026gatling} that retain predictable validity operate under precisely the multi-view-leader regime considered here, but without erasure coding, so that the leader's bandwidth remains a throughput bottleneck.

\vspace{0.2cm}
\noindent \textbf{Paper structure.}  Section~\ref{setup} describes the model and defines the key building blocks, including erasure codes and Merkle trees, and the formal definitions of SMR and Extractable SMR. Section~\ref{intu} provides an informal overview of the ideas behind our protocols.  Section~\ref{formal} gives the formal specification  of Carnot~1: analysis in  the standard model of partial synchrony appears in Appendix~\ref{anal1}.   Section \ref{2int} describes the intuition behind Carnot~2. Section \ref{2spec} gives the formal specification of Carnot~2, while Appendix \ref{anal3} presents analysis  in the standard model of partial synchrony. Section \ref{anal2} then analyses Carnot~2 in the Pipes model~\cite{lewis2025pipes}, which considers network bandwidth, and so allows for a formal analysis of throughput and block times.  Section~\ref{imposs} presents the impossibility result for 2-round finality. Section \ref{exper} describes the results of our experiments, which compare latency and throughput for Carnot 1 and 2 against the state-of-the-art.  Section \ref{rw} describes related work, and Section \ref{disc} contains a final discussion. 

%
%

\section{The Setup} \label{setup}

We consider a set $\Pi = \{p_1, \ldots, p_{n} \}$ of $n$ processors. At most $f$ processors may become corrupted by the adversary during the course of the execution (i.e., the adversary is \emph{adaptive}), and may then display  \emph{Byzantine} (arbitrary) behaviour.
Processors that never become corrupted by the adversary are referred to as \emph{correct}. While discussing Carnot 1, we assume $n\geq 4f+1$. While discussing Carnot 2,  we assume $n\geq 3f+1$. 

\vspace{0.2cm}
\noindent \textbf{Cryptographic assumptions}. Our cryptographic assumptions are standard for papers on this topic. Processors communicate by point-to-point authenticated channels. We use a cryptographic signature scheme, a public key infrastructure (PKI) to validate signatures, and a collision resistant hash function $H$. As expanded upon later in this section, we also  use threshold signatures and erasure codes.

  We assume a computationally bounded adversary. Following a common standard in distributed computing and for simplicity of presentation (to avoid the analysis of negligible error probabilities), we assume these cryptographic schemes are perfect, i.e., we restrict attention to executions in which the adversary is unable to break these cryptographic schemes.

  \vspace{0.2cm}
\noindent \textbf{The partial synchrony model}. As noted above, processors communicate using point-to-point authenticated channels. We consider the standard partial synchrony model \cite{DLS88}, whereby the execution is divided into discrete timeslots $t\in \mathbb{N}_{\geq 0}$ and a message sent at time $t$ must arrive at time $t'>t$ with $t'\leq \max\{\text{GST},t\} + \Delta$. While $\Delta$ is known, the value of GST is unknown to the protocol. We also write $\delta$ to denote the (unknown) least upper bound on message delay after GST (noting that $\delta$ may be significantly less than the known bound $\Delta$). The adversary chooses GST and also message delivery times, subject to the constraints already defined.

Correct processors begin the protocol execution before GST and are not assumed to have synchronised clocks.  For simplicity, we do assume that
the clocks of correct processors all proceed in real time,
meaning that if
$t'>t$ then the local clock of correct $p$ at time $t'$ is $t'-t$ in
advance of its value at time $t$. Using standard arguments, our
  protocol and analysis can
be extended in a straightforward way to the case
in which there is a known upper bound on the difference
between the clock speeds of correct processors.

\vspace{0.2cm}
\noindent \textbf{Threshold signatures}. A $k$-of-$n$ threshold signature scheme allows \emph{signature shares} from any $k$ processors on a given message to be combined to form a \emph{certificate} on that message. Forming a certificate on any message is infeasible given less than $k$ signature shares.  Such schemes can be implemented using BLS signatures \cite{boneh2001short}.    

\vspace{0.2cm}
\noindent \textbf{Erasure codes}.
 For each $k\in [1,n]$, we suppose given an $(n,k)$-erasure code, which uniquely encodes any bit string $C$ of length $\beta$ as a sequence of $n$ \emph{fragments}, $c_1,\dots,c_n$,  in such a way that any $k$  fragments and $\beta$ suffice to efficiently reconstruct $C$. We suppose all fragments have the same size (as a function of $n,k$ and $\beta$). 
Reed-Solomon codes can be used to realise an $(n,k)$-erasure code so that each fragment has size $\approx \beta/k$. For example, if  $n=3f+1$ and $k=2f$, this leads to a data expansion rate of roughly $1.5$, i.e., the combined size of all $n$ fragments is roughly 1.5$\beta$. The data expansion rate approaches 1 as $k$ approaches $n$.  

\vspace{0.2cm}
\noindent \textbf{Merkle trees}. We use Merkle trees in the standard way to allow a processor $p$ to commit to a sequence of values $v_1,\dots,v_n$. To form the commitment, $p$ constructs a full binary tree in which the leaves are the hashes of $v_1,\dots,v_n$ and every other node is the hash of its two children. The commitment is the root of the tree, $r$ say. 
To \emph{open} the commitment at position $i$, $p$ specifies $v_i$ and a \emph{validation path  from $r$ to $v_i$ at position $i$}: the validation path specifies the sibling of each node on the path from the hash of $v_i$ at position $i$ to the root. 

\vspace{0.2cm}
\noindent \textbf{Encoding, certified fragments, and tags}. We use techniques introduced by Cachin and Tessaro \cite{cachin2005asynchronous} for the purpose of \emph{asynchronous verifiable information dispersal (AVID)}. These techniques combine the use of Merkle trees and erasure codes. Given $k\in [1,n]$ and  a bit string $C$ of length $\beta$, let $c_1,\dots,c_n$ be the corresponding fragments produced by our $(n,k)$-erasure code. Form a Merkle tree whose leaves are the hashes of $c_1,\dots, c_n$ and let $r$ be the root of this tree. For each $i\in [n]$, let $\pi_i$ be a validation path  from $r$ to $c_i$ at position $i$. We define the \emph{tag} $\tau(C,k):= (\beta, k,r)$ and set: 
\[ \text{Encode}(C,k):= (\tau(C,k), \{ (c_i,\pi_i ) \}_{i\in [n]} ). \]

\noindent If $z= (\beta', k,r')$ for some $\beta'\in \mathbb{N}$ and some hash value $r'$, we say $(c,\pi)$ is a \emph{certified fragment of $z$ at $i$} if both: 
\begin{itemize} 
\item $c$ is of the correct length (given $n$ and $k$) to be a fragment of a message of length $\beta'$, and;
\item $\pi$ is a validation path from $r'$ to $c$ at position $i$. 
\end{itemize}

\vspace{0.2cm}
\noindent \textbf{Decoding}. The function Decode takes inputs of the form
\[ (z,  \{ (c_i,\pi_i ) \}_{i\in I} ), \]
where $z=(\beta, k,r)$ for $\beta\in \mathbb{N}$, $k\in [1,n]$,  $r$ is a hash value, $I\subseteq [n]$ with $|I|=k$, and each $(c_i,\pi_i)$ is a certified fragment of $z$ at $i$. It then reconstructs a message $C$ of length $\beta$ from the fragments $\{ c_i \}_{i\in I}$. If this reconstruction fails (e.g., due to a formatting error) it outputs $\bot$. Otherwise, it computes $\tau(C,k)=(\beta, k, r')$. If $r\neq r'$ it outputs $\bot$, and otherwise outputs $C$.

\vspace{0.2cm}
\noindent \textbf{Transactions}. Transactions are messages of a
distinguished form,  signed by the \emph{environment}. Each timeslot, each processor may receive some
finite set of transactions directly from the environment. We make the standard assumption that transactions are unique (repeat transactions can be produced using an increasing `ticker'  or timestamps \cite{castro1999practical}).

\vspace{0.2cm}
\noindent \textbf{SMR and Extractable SMR: informal discussion}.  Many protocols for State Machine Replication (SMR)~\cite{schneider1990implementing} do not \emph{explicitly} specify the mechanism by which correct processors retrieve the full finalised log of transactions. Instead, they ensure \emph{data availability} for the log. In Narwhal~\cite{danezis2022narwhal}, for example, blocks of transactions are sent to all processors, who then send acknowledgement messages confirming receipt. A set of $f+1$ such acknowledgements acts as a \emph{data availability certificate} for the block: it proves that at least one correct processor has received the block, and can therefore supply it to others upon request. Consensus can then be run on certificates rather than full blocks. 
The notion of \emph{Extractable SMR} was introduced in~\cite{lewis2025morpheus} to formalise the task being solved by such protocols. Roughly speaking, the requirement is that, while an individual processor may not be able to determine the full finalised log (without extra communication), the correct processors collectively receive sufficient messages to determine it. In short, the difference between SMR and Extractable SMR is that  SMR requires every correct processor to eventually hold the full finalised log, while Extractable SMR requires only that the log be recoverable from the information held by correct processors collectively.

\vspace{0.2cm}
\noindent \textbf{SMR: definition}.  Write $\sigma \preceq \tau$ to denote that the string $\sigma$ is a prefix of the string $\tau$. 
  If $\mathcal{P}$ is a protocol for \emph{SMR}, then it must specify a function $\mathcal{F}$ that maps any set of messages to a sequence of transactions. Let $M^*$ be the set of all messages that are received by at least one (potentially Byzantine) processor during the execution. For any timeslot $t$, let $M_{p_i}(t)$ be the set of all messages that are received by $p_i$ at timeslots  $\leq t$. We require the following conditions to hold:

\vspace{0.1cm} 
\noindent \emph{Consistency}.  For any $M_1\subseteq M_2  \subseteq M^*$,  $\mathcal{F}(M_1)\preceq \mathcal{F}(M_2)$.

\vspace{0.1cm} 
\noindent \emph{Liveness}.  If $p_i$ and $p_j$ are correct and if $p_i$ receives the transaction $\text{tr}$ then, for  some  $t$,  $\text{tr}\in \mathcal{F}(M_{p_j}(t))$. 

\vspace{0.1cm} 
At any time-slot $t$,  $\text{log}_i(t):=\mathcal{F}(M_{p_i}(t))$ is the sequence of transactions \emph{finalised} by $p_i$. 
This definition of consistency ensures that correct
processors never finalise incompatible sequences: for any sets of messages 
$M_1,M_2\subseteq M^*$ that two such processors might have received,
$ \mathcal{F}(M_1) \preceq \mathcal{F}(M_1 \cup M_2)$ and
$\mathcal{F}(M_2)\preceq \mathcal{F}(M_1\cup M_2)$.  

\vspace{0.2cm}
\noindent \textbf{Extractable SMR: definition}.  
  If $\mathcal{P}$ is a protocol for \emph{Extractable SMR}, then (as for SMR)  it must specify a function $\mathcal{F}$ that maps any set of messages to a sequence of transactions. Let $M^*$ be defined as in the definition of SMR.  For any timeslot $t$, let $M_c(t)$ be the set of all messages that are received by at least one correct processor at a timeslot $\leq t$. We require the following conditions to hold:

\vspace{0.1cm} 
\noindent \emph{Consistency}.  For any $M_1\subseteq M_2  \subseteq M^*$,  $\mathcal{F}(M_1)\preceq \mathcal{F}(M_2)$.

\vspace{0.1cm} 
\noindent \emph{Liveness}. If correct $p_i$ receives the transaction $\text{tr}$, there must exist $t$ such that $\text{tr}\in \mathcal{F}(M_c(t))$. 

\vspace{0.2cm}
\noindent \textbf{Blocks, parents, ancestors, and descendants}. We specify a protocol that produces \emph{blocks}.   There is a unique \emph{genesis} block $b_{\text{gen}}$, which is considered \emph{finalised} at the start of the protocol execution. Each block $b$ other than the genesis block has a unique \emph{parent}. The \emph{ancestors} of $b$ are $b$ and all ancestors of its parent (while the genesis block has only itself as ancestor), and each block has the genesis block as an ancestor.  Block $b'$ is a \emph{descendant} of $b$ if $b$ is an ancestor of $b'$. Two blocks are \emph{inconsistent} if neither is an ancestor of the other.

\section{Carnot: the intuition} \label{intu} 

In this section, we informally describe the intuition behind our positive results. In Section \ref{imposs}, we show that achieving a data expansion rate lower than 2.5 is not possible for protocols with 2-round finality. In this section, we therefore consider protocols with two rounds of voting (3-round finality). We start by supposing $n\geq 3f+1$. First, in Section \ref{simplex}, we review the Simplex protocol~\cite{chan2023simplex}, upon which our protocols will be based. Then, in Section \ref{sl}, we briefly consider the throughput advantages of multi-view leaders. In Sections \ref{ber}-\ref{lower}, we consider increasingly advanced approaches to reducing data expansion.

\subsection{Recalling Simplex} \label{simplex} 
Simplex~\cite{chan2023simplex} is a  protocol for partial synchrony assuming $n \geq 3f+1$, with rotating leaders. It proceeds in sequential \emph{views} $v = 1, 2, \ldots$, each with leader $\mathtt{lead}(v)$.

\vspace{0.2cm} 
\noindent \textbf{Two rounds of voting}. The protocol uses two rounds of voting per view. When a processor enters view $v$, it sets a timer to expire in time $3\Delta$. When the leader enters view $v$, it immediately proposes a new block.   Upon receiving a valid proposal $b$, each processor disseminates (sends to all) a \emph{stage-1} vote for $b$. A set of $n-f$ stage-1 votes for $b$, each by a different processor, is called a \emph{stage-1 notarisation} for $b$. Upon receiving a stage-1 notarisation for $b$, a processor moves to view $v+1$. If, upon entering $v+1$, the processor's timer for view $v$ has not yet fired, it disseminates a \emph{stage-2} vote for $b$. A set of $n-f$ stage-2 votes is called a \emph{stage-2 notarisation} for $b$, and $b$ is \emph{finalised} when it receives such a notarisation. 

\vspace{0.2cm} 
\noindent \textbf{Nullifications}. If the leader is faulty or slow, processors time out after $3\Delta$ and disseminate a \emph{nullify}$(v)$ message. A set of $n-f$ nullify$(v)$ messages is called a \emph{nullification} for view $v$; processors also move to view $v+1$ upon receiving a nullification. The critical design constraint is that each correct processor either sends a stage-2 vote or a nullify$(v)$ message, but never both. It follows by the standard quorum intersection argument ($(n-f)+(n-f)-n\geq f+1$ when $n\geq 3f+1$) that a stage-2 notarisation and a nullification for the same view cannot coexist. Similarly, two distinct blocks cannot both receive a stage-1 notarisation in the same view.

\vspace{0.2cm}
\noindent \textbf{Progression through views}. Correct processors disseminate stage-1 notarisations and nullifications upon first receipt. If any correct processor leaves view $v$ upon receiving a stage-1 notarisation, then all correct processors receive it and also leave the view. Otherwise, all correct processors eventually time out and disseminate nullify$(v)$ messages, producing a nullification.

\vspace{0.2cm} 
\noindent \textbf{Consistency}. We stipulated above that correct processors will disseminate stage-1votes upon receipt of a \emph{valid} proposal from the leader. To ensure consistency, we further stipulate that $p_i$ will only regard a view $v$ block with parent $b'$ for view $v'$ as a valid proposal if it has received a stage-1 notarisation for $b'$ and nullifications for all views in the open interval $(v',v)$. We noted above that, if  a view $v$ block $b$ is finalised, then view $v$ does not receive a nullification. This means that correct processors cannot vote for any block inconsistent with $b$ in a subsequent view, since no such proposal could be valid in the absence of a nullification for view $v$.  

\vspace{0.2cm} 
\noindent \textbf{Liveness}. Recall that $\delta$ denotes the (unknown) least upper bound on message delay after GST. The key lemma is that views are synchronised: if the first correct processor to enter view $v$ does so at $t\geq \text{GST}$, then all correct processors enter view $v$ by $t+\delta$, because processors disseminate stage-1 notarisations and nullifications upon first receipt.

Suppose the leader of view $v$ is correct and proposes $b$ with parent $b'$ for view $v'$. The leader enters view $v$ and proposes by $t+\delta$. Since the leader has entered view $v$, it must have received nullifications for all views in $(v',v)$ as well as a stage-1 notarisation for $b'$; it disseminates all of these upon first receipt. All correct processors therefore receive these, together with the proposal $b$, by $t+2\delta$. They send stage-1 votes by $t+2\delta$ and receive a stage-1 notarisation for $b$ by $t+ 3\delta$. Since $\delta \leq \Delta$, this occurs before any timer fires (the earliest a timer can fire is $t + 3\Delta \geq t+3\delta$), so all correct processors send stage-2 votes upon entering view $v+1$, and $b$ is finalised by $t+4\delta$.

\subsection{The advantage of multi-view leaders} \label{sl}

\noindent \textbf{The cost of rotating leaders}. In Simplex, each view has a different leader. Before proposing, the new leader must receive a stage-1 notarisation or nullification from the previous view, and this creates a significant gap between consecutive proposals during which no block-data is being disseminated, limiting throughput and increasing latency. One can partially address this with \emph{optimistic proposals}~\cite{doidge2024moonshot}, whereby the leader of view $v+1$ proposes a block upon receiving the proposal for view $v$, without waiting for a stage-1 notarisation. This reduces the gap, but does not eliminate it: the new leader must still receive the previous leader's proposal before it can propose its own block.

\vspace{0.2cm}
\noindent \textbf{Multi-view leaders and pipelining}. A more direct solution is to give the same leader multiple successive views. If a correct leader has $x$ consecutive views, it can pipeline proposals: as soon as it sends the block for view $v$, it can immediately begin sending the block for view $v+1$, without waiting for any messages. This eliminates the inter-proposal gap entirely for all but the first view of each leader's tenure, allowing throughput to approach the underlying network bandwidth as $x$ grows. In our protocol, we formalise this using \emph{superviews}: each superview is a sequence of $x$ views with a single leader, and leaders rotate between superviews.

\subsection{Basic erasure coding}  \label{ber} 

We now turn to the question of how each individual block is disseminated. In Simplex, the leader sends a copy of the full block to every other processor, so that the total data sent is approximately $n$ times the block size, giving a data expansion rate of $n$. A natural approach to reducing this is to use erasure coding. In this section, we describe a straightforward approach producing a data expansion rate of 3, using methods first described in~\cite{cachin2005asynchronous}. In Sections~\ref{1.5} and~\ref{lower}, we consider how to achieve lower data expansion rates.

\vspace{0.2cm}
\noindent \textbf{Erasure-coded proposals}. Rather than sending the entire block to each processor, the leader separates out the block's \emph{payload} $C$ (its sequence of transactions) and encodes it using an $(n,n-2f)$-erasure code, producing $n$ fragments of size approximately $|C|/(n-2f)$ each, together with a short commitment (the \emph{tag} $\tau(C,n-2f)$, as defined in Section~\ref{setup}). The block $b$ itself is now just a small tuple $(v, \tau(C,n-2f), h)$ signed by the leader, where $h$ is the hash of the parent block. The leader sends each processor $p_i$ a single \emph{certified fragment} $(b, i, c_i, \pi_i)$, where $c_i$ is $p_i$'s fragment of the payload and $\pi_i$ is a validation path allowing $p_i$ to verify that its fragment is consistent with the tag. The total data sent by the leader is approximately $n/(n-2f)$ times the payload size.

\vspace{0.2cm}
\noindent \textbf{Voting and fragment dissemination}. Upon receiving and verifying its certified fragment, processor $p_i$ disseminates both the fragment and a stage-1 vote for $b$. If $n-f$ stage-1 votes are collected, a stage-1 notarisation is formed. To proceed safely to the next view, however, a processor must not only receive a stage-1 notarisation but also collect enough fragments to reconstruct the payload.

\vspace{0.2cm}
\noindent \textbf{Why $n-2f$ fragments must suffice}. Consider a correct processor $p_i$ that receives a stage-1 notarisation for $b$. Since $p_i$ disseminates the notarisation upon receipt, all correct processors will eventually receive it too. But can all correct processors also reconstruct the payload? The notarisation contains $n-f$ votes, of which at least $n-2f$ are from correct processors. These correct voters each disseminate their fragment, so every correct processor is guaranteed to receive at least $n-2f$ distinct fragments. With an $(n,n-2f)$-erasure code, this suffices for reconstruction, ensuring that all correct processors can reconstruct the payload and proceed to the next view. (Note that $p_i$ itself may have received more than $n-2f$ fragments. For example, if the leader is correct, $p_i$ will eventually receive fragments from all  correct processors. The point is that $n-2f$ is the number we can guarantee for \emph{every} correct processor.)

\vspace{0.2cm}
\noindent \textbf{The data expansion rate is 3}. The data expansion rate is $n/(n-2f)$, which equals $(3f+1)/(f+1)$ when $n=3f+1$, approaching 3 as $f$ grows.

\vspace{0.2cm}
\noindent \textbf{Stage-2 voting and finalisation}. Once a processor has received a stage-1 notarisation and reconstructed the payload, it proceeds to the next view and, if its timer has not fired, sends a stage-2 vote, exactly as in Simplex. The safety and liveness arguments carry over from Section~\ref{simplex} without essential modification: the key properties --- that each correct processor sends at most one stage-1 vote and either a stage-2 vote or a nullification per view --- are unchanged.

\vspace{0.2cm}
\noindent \textbf{Threshold certificates}. In the description above, a stage-1 notarisation is a set of $n-f$ individual votes, each of which must be disseminated. In practice, one can reduce the communication overhead by using threshold signatures: each processor's vote includes a signature share, and $n-f$ shares can be combined into a single compact \emph{threshold certificate} that any processor can verify. 

\subsection{Achieving a data expansion rate of 1.5}  \label{1.5} 

The data expansion rate of 3 arose because a stage-1 notarisation only guarantees that $n-2f$ fragments reach every correct processor, forcing the use of an $(n,n-2f)$-erasure code. To reduce the data expansion rate, we need to increase the number of fragments required for reconstruction. In this section, we describe how to achieve a data expansion rate of $n/(n-f)$, which equals $(3f+1)/(2f+1)$ when $n=3f+1$, approaching $1.5$ as $f$ grows. 

\vspace{0.2cm}
\noindent \textbf{Using an $(n,n-f)$-erasure code}. Blocks are formed using erasure coding as before, but now with an $(n,n-f)$-erasure code: $n-f$ fragments are required for reconstruction. A processor $p_i$ proceeds to view $v+1$ upon receiving both a stage-1 notarisation for a view $v$ block $b$ and fragments of $b$ from $n-f$ distinct processors. Under standard operation, i.e., correct leaders during synchrony, these $n-f$ fragments arrive naturally as part of stage-1 voting, and the protocol behaves just as in Section~\ref{ber}.

\vspace{0.2cm}
\noindent \textbf{The problem: ensuring data availability}. The difficulty arises when the leader is faulty. In such cases, a correct processor $p_i$ might collect $n-f$ fragments and form a stage-1 notarisation, but some other correct processor $p_j$ might not receive enough fragments to reconstruct the payload. Since we now require $n-f$ fragments rather than $n-2f$, the guarantee from Section~\ref{ber} (that the correct voters alone provide enough fragments) no longer suffices. We therefore need an additional mechanism to ensure that all correct processors eventually receive the fragments they need.

\vspace{0.2cm}
\noindent \textbf{Fragment echoing}. One solution is as follows. Upon proceeding to view $v+1$, processor $p_i$ continues with the instructions for the new view but also, in the background (without slowing the critical path), sets a timer.  When the timer expires, $p_i$ checks which processors have not yet sent their fragment of $b$, and sends each such processor $p_j$ their fragment (i.e., $p_j$'s fragment, which $p_i$ can reconstruct from the decoded payload)  along with the stage-1 notarisation for $b$. Additionally, any processor $p_j$ that receives its own fragment of $b$ together with a stage-1 notarisation (from any processor, not just the leader) is required to \emph{echo} that fragment by disseminating it to all processors, even if $p_j$ has since moved to a later view or has previously voted for a different block. The stage-1 notarisation requirement here is important: it limits the amount of echoing that Byzantine processors can provoke, since extra echoing is triggered only for a block that has actually been notarised.

\vspace{0.2cm}
\noindent \textbf{Two types of extra communication}. Under standard operation, no extra fragment echoing occurs. When it does occur, there are two types: (a)~$p_i$ may have to send $p_j$ their fragment (if $p_i$ has not received that fragment from $p_j$); (b)~$p_j$ may subsequently have to echo their fragment to all processors upon receiving it together with the stage-1 notarisation.


\subsection{Pushing the data expansion rate lower than 1.5}  \label{lower} 

While a data expansion rate of $1.5$ may be necessary when $f$ processors are actually faulty, when fewer processors are faulty it should be possible for leaders to use lower expansion rates. The idea is to allow leaders to \emph{attempt} lower expansion rates and revert to a safe expansion rate if reconstruction fails. 

\vspace{0.2cm}
\noindent \textbf{The difficulty}. When the leader uses an $(n,k)$-erasure code with $k > n-f$, the $n-f$ fragments echoed by correct voters may not suffice for reconstruction. If \emph{no} correct processor can reconstruct the payload, this is unproblematic: all correct processors eventually time out and nullify the view. The difficulty arises when \emph{some} correct processors reconstruct the payload (perhaps with the aid of fragments from Byzantine processors) while others cannot. In this case, unlike in Section~\ref{1.5}, it does not suffice to have correct processors echo their fragments: even if all correct processors receive and echo their fragments, the $n-f$ fragments from correct processors may not be enough for reconstruction when $k > n-f$.

 \vspace{0.2cm}
\noindent \textbf{Carnot 2: an approach for $n\geq 3f+1$}. Perhaps the most natural solution, which we present as Carnot~2, is as follows. When a processor decodes the payload and moves to the next view, but does not receive timely  assurance that all correct processors can do likewise (the precise conditions triggering this fallback are specified in Section~\ref{2spec}), it reverts to encoding the payload using an $(n,n-f)$-erasure code\footnote{We actually use an $(n,n-f-1)$-erasure code, so that the leader does not need to echo their own fragment.}  and sends out the relevant fragments. Upon receiving their fragment together with the corresponding stage-1 certificate, other processors echo their fragment to all. In this way, if any correct processor leaves the view upon decoding the block, it can be sure all others will eventually recover it. While this approach works for $n\geq 3f+1$, it has the drawback that Byzantine action can trigger  extra fragment echoing.

\vspace{0.2cm}
\noindent \textbf{Carnot 1: avoiding extra echoing for $n\geq 4f+1$}. Before presenting Carnot 2, we present a simple approach that assumes $n\geq 4f+1$\footnote{Note that the approach of Section~\ref{1.5} already achieves a data expansion rate of $n/(n-f) = (4f+1)/(3f+1)$ when $n=4f+1$, which approaches $4/3 \approx 1.33$ as $f$ grows. The techniques in this section are concerned with pushing the rate below this value.} but avoids \emph{any} extra fragment echoing: processors echo their certified fragment once upon receiving it from the leader and voting, but no further echoing is ever required. The key is to note that  if $f+1$ stage-2 votes for $b$ are collected, at least one must be from a correct processor, and a correct processor only sends a stage-2 vote after successfully reconstructing the payload. The $f+1$ votes can therefore be combined into a compact threshold certificate --- which we call a stage-2 \emph{M-certificate} (where `M' stands for `mini') --- that serves as a certificate of data availability. Since any block receiving an M-certificate must also have received a stage-1 notarisation (because each correct stage-2 voter must first have seen one), an M-certificate on its own suffices for view progression: a processor that receives an M-certificate for a view $v$ block can proceed to view $v+1$ without needing to collect $n-f$ fragments directly. This eliminates the need for further fragment echoing upon receipt of an M-certificate.
We use stage-2 notarisations of size $n-f$ and nullifications of size $2f+1$. An  \emph{N-certificate} is a threshold certificate formed from a nullification. Two properties follow.

\vspace{0.2cm}
\noindent \textbf{(a) A stage-2 notarisation and an N-certificate for the same view cannot coexist}. This follows from the standard quorum intersection argument: $(n-f) + (2f+1) - n = f+1$, so any stage-2 notarisation and any nullification for the same view must share at least one correct processor. Since no correct processor sends both a stage-2 vote and a nullify message for the same view, such a pair cannot form.

\vspace{0.2cm}
\noindent \textbf{(b) Each view produces an M-certificate or an N-certificate (or both)}. Recall that an M-certificate requires only $f+1$ stage-2 votes. If more than $f$ correct processors send stage-2 votes for view $v$, then, since all stage-2 votes by correct processors are for the same block, at least $f+1$ stage-2 votes for the same block are collected, producing an M-certificate. If at most $f$ correct processors send stage-2 votes, then at least $n-2f \geq 2f+1$ correct processors send nullify$(v)$ messages, producing an N-certificate.

\vspace{0.2cm}
\noindent \textbf{Consequences for view progression}. Property~(a) guarantees Consistency will still be satisfied.  Property~(b) means we can require either an M-certificate  or an N-certificate for view progression.  Since an M-certificate guarantees that at least one correct processor has decoded the payload, data availability is ensured without any additional fragment echoing. In Section \ref{formal}, we will implement this using the following terminology: each processor maintains a local variable $\mathtt{blocks}$, which  stores all blocks for which it has received an M-certificate. A  processor advances to view $v+1$ only  upon adding a view $v$ block to its local value $\mathtt{blocks}$,  or else upon receiving an N-certificate for view $v$.


\vspace{0.2cm}
\noindent \textbf{Carnot 1 vs Carnot 2}. To summarise, Carnot~2 works under the weaker assumption $n\geq 3f+1$ and solves SMR, but requires extra fragment echoing when Byzantine processors interfere. Carnot~1 assumes $n\geq 4f+1$ and solves Extractable SMR, but is simpler and requires no extra fragment echoing beyond the initial echo upon voting.

\section{Carnot 1: The formal specification} \label{formal} 
 We say `disseminate' to mean `send to all processors'. When a correct processor is instructed to send a message to itself, it regards that message as immediately
received.   The pseudocode uses a number of message types, local
variables, predicates, functions and procedures, detailed below. 

\vspace{0.2cm}
\noindent  \textbf{Superviews and the parameter $x$}. Each superview is a sequence of $x$ views.  For $w\in \mathbb{N}_{\geq 1}$, \emph{superview} $w$ is the set of views in $[(w-1)x+1, (w-1)x+x]$. If $v \equiv1 \text{ mod }x$, then $v$ is called an \emph{initial view}. 

\vspace{0.2cm}
\noindent  \textbf{The function} $\mathtt{lead}(v)$.  If view $v$ belongs to superview $w$, we set $\mathtt{lead}(v):= p_{j+1}$, where $j \equiv w \text{ mod }n$. So, all views in a superview have the same leader, while leaders for superviews rotate.

\vspace{0.2cm}
\noindent \textbf{Blocks}. Recall that $\tau$ is the tag function, as defined in Section \ref{setup}. The \emph{genesis block} is the tuple $b_{\text{gen}}:=(0,\tau(\lambda,n),\lambda)$, where $\lambda$ denotes the empty sequence (of length 0).  For some\footnote{As noted previously, we use the interval $[n-f-1,n-1]$, rather than $[n-f,n]$, so that leaders do not have to disseminate their own fragments.} $k\in [n-f-1,n-1]$, a block other than the genesis block, with associated payload $C$,  is a tuple $b=(v,\tau(C,k), h)$ signed by $\mathtt{lead}(v)$, where:
\begin{itemize}
\item $v\in \mathbb{N}_{\geq 1}$ is the view corresponding to $b$;
\item $\tau(C,k)$ is the tag resulting from $C$ and $k$;
\item $h$ is the hash of $b$'s parent block.
\end{itemize}
We write $b.\text{view}$, $b.\text{tag}$ and $b.\text{par}$ to denote the corresponding entries of $b$, and refer to $k$ as the \emph{reconstruction parameter} for $b$.  If $b.\text{view}=v$, we also refer to $b$ as a `view $v$ block' or a `block for view $v$'. If view $v$ is in superview $w$, we may also refer to $b$ as a block for superview $w$.  If $(c_i,\pi_i)$ is a certified fragment of $\tau(C,k)$ at $i$, we also say that the tuple $(b,i,c_i,\pi_i)$ is \emph{certified fragment of} $b$ at $i$: leaders will disseminate messages of this form when they propose a block. 

\vspace{0.2cm}
\noindent \textbf{Votes}. A  \emph{stage-$1$ vote} by $p_i\in \Pi$ for the block $b$ is a message of the form $(\text{vote},b,1,i)$, signed by $p_i$.   A  \emph{stage-$2$ vote} by $p_i\in \Pi$ for the block $b$ is a message of the form $(\text{vote},b,2,i,\rho_i,\rho_i')$, where: 
\begin{itemize} 
\item $\rho_i$ is a (valid) signature share from $p_i$ on the message $(\text{vote},b,2)$, using an $(f+1)$-of-$n$ threshold signature scheme; 
\item $\rho_i'$ is a signature share from $p_i$ on the message $(\text{vote},b,2)$, using an $(n-f)$-of-$n$ threshold signature scheme. 
\end{itemize}

\vspace{0.2cm}
\noindent \textbf{Notarisations and certificates}. These are defined as follows: 
\begin{itemize} 
\item A  \emph{stage-$1$ notarisation} for the block $b$ is a set of at least $n-\frac{3}{2}f$ stage-$1$ votes\footnote{Where $\frac{3}{2}f$ is not an integer, this means a notarisation requires $\lceil n-\frac{3}{2}f \rceil$ votes.} for $b$, each by a different processor in $\Pi$. (By a \emph{stage-$1$ notarisation}, we mean a stage-$1$ notarisation for some block.) 

\item A  \emph{stage-$2$ M-notarisation} \footnote{ `M' stands for `mini'.} for the block $b$  is a set of $f+1$ stage-$2$ votes for $b$, each by a different processor in $\Pi$.
A \emph{stage-$2$ M-certificate} for the block $b$ is the message $(\text{M-Cert},b,\rho)$, where $\rho$ is an $(f+1)$-of-$n$ threshold certificate on the message $(\text{vote},b,2)$.\footnote{We also say `M-certificate' to mean `stage-2 M-certificate'. } 

\item A  \emph{stage-$2$ notarisation} for the block $b$   is a set of $n-f$ stage-$2$ votes for $b$, each by a different processor in $\Pi$.
A \emph{stage-$2$ certificate} for the block $b$ is the message $(\text{Cert},b,\rho)$, where $\rho$ is an $(n-f)$-of-$n$ threshold certificate on the message $(\text{vote},b,2)$. 

\end{itemize} 

\vspace{0.2cm}
\noindent \textbf{Nullifications and N-certificates}. For $v\in \mathbb{N}_{\geq 1}$, a nullify$(v)$ message by $p_i$ is of the form $(\text{nullify},v,i,\rho_i)$, where $\rho_i$ is a (valid) signature share from $p_i$ on the message $(\text{nullify},v)$, using a $(2f+1)$-of-$n$ threshold signature scheme. A \emph{nullification} for view $v$ is a set of $2f+1$ nullify$(v)$ messages, each by a different processor in $\Pi$. (By a \emph{nullification}, we mean a nullification for some view.) An \emph{N-certificate} for view $v$ is a message $(\text{N-cert}, v,\rho)$, where $\rho$ is a $(2f+1)$-of-$n$ threshold certificate on the message $(\text{nullify,}v)$.

\vspace{0.2cm}
\noindent \textbf{The local variable} $\mathtt{v}$. Initially set to 1, this variable specifies the present view of a processor.

\vspace{0.2cm}
\noindent \textbf{The local variable} $\mathtt{b}$. Initially set to $b_{\text{gen}}$, this variable is used by leaders to choose a parent block. 

\vspace{0.2cm}
\noindent \textbf{The local variable} $\mathtt{S}$. This variable is maintained locally by each processor $p_i$ and stores all messages received. It is considered to be automatically updated, i.e., we do not give explicit instructions in the pseudocode updating $\mathtt{S}$.  If $p_i$ receives a stage-$2$ (M-)notarisation or a nullification, then it automatically forms the associated certificate, and adds it to $\mathtt{S}$. 

\vspace{0.2cm}
\noindent \textbf{The local timer} $\mathtt{T}$. Each processor $p_i$ maintains a local timer $\mathtt{T}$, which is initially set to 0 and increments in real-time. (Processors will be explicitly instructed to reset their timer to 0 upon entering a new view.)

\vspace{0.2cm}
\noindent \textbf{New certificates}. An N-certificate $Q$ for view $v$ is regarded as \emph{new} at timeslot $t$ if $Q\in \mathtt{S}$ and $ \mathtt{S}$ did not contain an N-certificate for view $v$ at any timeslot $t'<t$. Similarly, a stage-$2$ certificate/M-certificate $Q$ for $b$ is new at timeslot $t$ if $Q\in \mathtt{S}$ and $ \mathtt{S}$ did not contain a stage-$2$ certificate/M-certificate for $b$ at any timeslot $t'<t$.

\vspace{0.2cm}
\noindent \textbf{The local variable} $\mathtt{blocks}$. Initially set to $\{ b_{\text{gen}} \}$, this local variable is automatically updated by $p_i$, without explicit instructions in the pseudocode. At any timeslot, $\mathtt{blocks}$ contains $b_{\text{gen}}$ and all blocks $b$ such that both conditions below are satisfied: 
\begin{itemize}
\item[(i)] $\mathtt{S}$ contains a stage-2  M-certificate for $b$;
\item[(ii)] There exists $b'\in \mathtt{blocks}$ with $H(b')=b.\text{par}$, i.e., the parent of $b$ is in $\mathtt{blocks}$. 
\end{itemize} 

\vspace{0.2cm}
\noindent \textbf{The local variable} $\mathtt{blocks}^*$. Initially set to $\{ b_{\text{gen}} \}$, this local variable is automatically updated by $p_i$, without explicit instructions in the pseudocode. It is similar to $\mathtt{blocks}$, but is used to determine when $p_i$ can disseminate stage-2 votes. At any timeslot, $\mathtt{blocks}^*$ contains $b_{\text{gen}}$ and all blocks $b$ such that, for $k$ with $b.\text{tag}=(\beta,k,r)$,  the conditions below are all satisfied: 
\begin{itemize}
\item[(i)] Processor $p_i$ has received a stage-1 notarisation for $b$;
\item[(ii)] There exists $I\subset [n]$ with $|I|=k$ such that, for each $j\in I$, $p_i$ has received $(b,j,c_j,\pi_j)$ from $p_j$, which is a certified fragment of $b$ at $j$.  On input $(b.\text{tag}, \{ (c_j,\pi_j) \}_{j\in I})$, Decode does not output $\bot$;
\item[(iii)] There exists $b'\in \mathtt{blocks}$ with $H(b')=b.\text{par}$, i.e., the parent of $b$ is in $\mathtt{blocks}$. 
\end{itemize}

\vspace{0.2cm}
\noindent \textbf{The local variables} $\mathtt{nullified}(v)$, $1\mathtt{voted}(v)$, and  $2\mathtt{voted}(v)$. These are used by $p_i$ to record whether it has yet sent a nullify$(v)$ message and whether  it has sent stage-1 or stage-2 votes for a view $v$ block. All three values are initially set to false for all $v$. 

\vspace{0.2cm}
\noindent \textbf{The function $g$}.  This function is used by the leader of a view to determine the value of $k$ for erasure coding. We allow flexibility with respect to the precise definition of $g$ but, for the sake of concreteness, we suppose (for now) that the output is always in $[n-f-1,n-1]$ and is a function of the received messages $\mathtt{S}$ and the view  $v$. When proving liveness in Appendices \ref{anal1} and \ref{anal3}, we also suppose $g(\mathtt{S},v)=n-f-1$ when $v$ is an initial view (while the output may be greater for non-initial views).\footnote{We make this assumption to simplify  the proof of liveness. In practice, one may wish to set higher values of $k$ for initial views in some cases.}

\vspace{0.2cm}
\noindent \textbf{The procedure} ProposeBlock. This procedure is executed by a  leader $p_i$ to produce and send out a new block. To execute the procedure, $p_i$ proceeds as follows:
\begin{itemize}
\item Let $\mathtt{v}$ be as locally defined for $p_i$ and suppose view $\mathtt{v}$ is in superview $w$. 
\item If $\mathtt{v}$ is initial and $p_i$ has not previously proposed any blocks for superview $w$, set $v:=\mathtt{v}$ and $b^\ast:=\mathtt{b}$. The block $b^\ast$ will be the parent of the new proposed block.  Otherwise, let $v-1$ be the greatest view in superview $w$ such that $p_i$ has already proposed a view $v-1$ block $b'$, and set $b^\ast:=b'$. 
\item  Form a payload $C$, containing all transactions  received  but not known to be included in the payloads of ancestors of $b^\ast$.\footnote{In the standard model of partial synchrony, where a block fragment of any size arrives within time $\Delta$ of sending after GST, there is no need to limit the size of blocks. In realistic contexts where processors have limited bandwidth (as formalised by the Pipes model \cite{lewis2025pipes}), it may be necessary to limit the size of blocks to ensure liveness with fixed time-outs. This is discussed further in Section \ref{anal2}.}
\item Set $k:=g(\mathtt{S},v)$ and calculate Encode$(C,k):=(\tau(C,k), \{ (c_j,\pi_j ) \}_{j\in [n]} )$. 
\item Set $b:=(v,\tau(C,k),H(b^\ast))$ signed by $p_i$ and, for each $j\in [n] \setminus \{ i \}$, send $(b,j,c_j,\pi_j)$ to $p_j$, i.e., send the certified fragment of $b$ at  $j$ to $p_j$ (we say `$p_i$ proposes $b$'). 
\end{itemize}

\vspace{0.2cm}
\noindent \textbf{The local predicate $\mathtt{ProposeReady}$}.  This predicate is used to determine whether a leader $p_i$ is ready to propose a new block. Let $\mathtt{v}$ be as locally defined for $p_i$ and suppose view $\mathtt{v}$ is in superview $w$. There are two cases: 

\vspace{0.1cm} 
\noindent \emph{Case 1}: View  $\mathtt{v}$ is initial and $p_i$ has not previously proposed any blocks for superview $w$. In this case, $\mathtt{ProposeReady}=\text{true}$.  This means that a correct processors will propose a new block immediately upon entering any initial view for which they are leader. 

\vspace{0.1cm} 
\noindent \emph{Case 2}: Otherwise. In this case,  the idea is that (in a context where the time taken to send data depends on the amount of data sent) $\mathtt{ProposeReady}$ should  be true if and only if $p_i$ has `finished' sending the previous block (although it may not yet have been received by other processors) and has not already proposed blocks for all views in the superview.\footnote{In realistic scenarios, one may also wish to consider further conditions, such as having received at least a certain number of transactions to include in a new block.} These issues make sense in the context of the Pipes model \cite{lewis2025pipes}, where processors have limited bandwidth, and will be discussed further in Section \ref{anal2}. For now, working only in the standard model of partial synchrony, we allow flexibility as to how  $\mathtt{ProposeReady}$ should be defined in Case 2.   Formally, we  assume only that, if $p_i=\mathtt{lead}(v)$ is correct and if correct $p_j$ enters non-initial view $v+1$ at $t$ upon adding a view $v$ block to its local value $\mathtt{blocks}$, then  there must exist $t'<t$ such that the following holds: $\mathtt{ProposeReady}=\text{true}$ as locally defined for $p_i$ at $t'$ and  $p_i$ proposes a block for view $v+1$ at $t'$.

 \vspace{0.2cm}
\noindent \textbf{When $p_i$ has received a votable fragment for view $v$}. If $b=(v,\sigma,h)$ is a view $v$ block signed by $\mathtt{lead}(v)$ and $(b,i,c_i,\pi_i)$ is a certified fragment  of $b$ at $i$, then $p_i$ regards this certified fragment as \emph{votable for view $v$} if: 
\begin{enumerate}
\item[(i)] there exists $b'\in \mathtt{blocks}$ with $H(b')=h$, $b'.\text{view}<v$, and; 
\item[(ii)] $\mathtt{S}$ contains an N-certificate for each view in the open interval $(b'.\text{view},v)$.
\end{enumerate}

\vspace{0.2cm}
\noindent \textbf{The local predicate $\mathtt{TimeoutReady}$}.  This predicate is used to determine when $p_i$ should send a nullify$(\mathtt{v})$ message. It is set to true if any of the following conditions apply, and is otherwise false: 
\begin{enumerate} 
\item[(a)] $\mathtt{v}$ is not initial, $\mathtt{T}=\Delta$, and  $1\mathtt{voted}(\mathtt{v})=2\mathtt{voted}(\mathtt{v})=\text{false}$, or; 
\item[(b)] $\mathtt{v}$ is not initial, $\mathtt{T}=2\Delta$, and  $2\mathtt{voted}(\mathtt{v})=\text{false}$, or; 
\item[(c)] $\mathtt{v}$ is initial, $\mathtt{T}=2\Delta$, and  $1\mathtt{voted}(\mathtt{v})=2\mathtt{voted}(\mathtt{v})=\text{false}$, or;                 
\item[(d)] $\mathtt{v}$ is initial, $\mathtt{T}=3\Delta$, and  $2\mathtt{voted}(\mathtt{v})=\text{false}$:
\end{enumerate} 
\noindent To understand the time-outs above, note that if $v$ is non-initial, then, during synchrony, a correct leader for view $v$ will propose a view $v$ block before $p_i$ enters the view.  If $v$ is initial, then a correct leader for view $v$ will propose a view $v$ block within $\Delta$ time-slots of $p_i$ entering the view. Then $p_i$ must wait at most a further  $\Delta$ time-slots before sending their stage-1 vote, and then at most a further $\Delta$ time-slots before sending a stage-2 vote for the proposal. 

\vspace{0.2cm}
\noindent \textbf{The function $\mathcal{F}$ for Extractable SMR}. Given any set of messages $M$, let $b_1,\dots,b_m$ be the longest sequence of blocks such that, for each $i<m$, $b_{i}$ is the parent of $b_{i+1}$, and, for each $i\leq m$, all of the following conditions are satisfied: 
\begin{itemize} 
\item There exists $j\geq  i$ such that $M$ contains a stage-2 certificate for $b_j$. 
\item Let $b_i.\text{tag}=(\beta,k,r)$. There exists $I\subset [n]$ with $|I|=k$ such that, for each $j\in I$, $M$ contains $(b,j,c_j,\pi_j)$, which is a certified fragment of $b$ at $j$.  On input $(b.\text{tag}, \{ (c_j,\pi_j) \}_{j\in I})$, Decode does not output $\bot$;
\end{itemize} 
Then we define $\mathcal{F}(M)$ to be the sequence of transactions formed by concatenating the payloads of $b_1,\dots,b_m$. If there does not exist a unique longest sequence $b_1,\dots,b_m$ as specified above for $M$, then we set  $\mathcal{F}(M)$ to be the empty sequence.

\vspace{0.2cm} For ease of reference, local variables are displayed in the table below. The pseudocode appears in Algorithm 1. Appendix \ref{anal1} verifies Consistency and Liveness in the standard model of partial synchrony, and also discusses communication complexity.

\begin{table}[h!]
  \begin{center}
    \label{tab:table8}
    \begin{tabular}{l|l} 
      \textbf{Variable} & \textbf{Description} \\
      \hline
      $\mathtt{v}$ & Initially 1, specifies the present view \\
       $\mathtt{b}$ & Initially $b_{\text{gen}}$, used to specify parents \\
      $\mathtt{T}$ & Initially 0, a local timer reset upon entering each view \\
$\mathtt{nullified}(v)$ & Initially false, specifies whether already sent nullify$(\mathtt{v})$ message \\
$1\mathtt{voted}(v)$ & Initially  $\text{false}$, records whether stage-1 vote already disseminated \\
$2\mathtt{voted}(v)$ & Initially $\text{false}$, records whether stage-2 vote already disseminated \\

 $\mathtt{blocks}$ & Initially contains only $b_{\text{gen}}$, records all blocks with stage-2 M-certificates \\
 &  Automatically updated \\
  $\mathtt{blocks}^*$ & Initially contains only $b_{\text{gen}}$, records blocks with stage-1 notarisations \\
 &  and recovered payloads. Automatically updated \\
 $\mathtt{S}$ & Records all received messages, automatically updated \\
& Initially contains only $b_{\text{gen}}$  \\

    \end{tabular}
        \caption{Local variables for Carnot 1}

  \end{center}
\end{table}

\setcounter{algorithm}{0}

 \begin{algorithm}
\caption{: Carnot 1, the instructions for $p_i$.}
 \label{alg1}
\begin{algorithmic}[1]


\State  At every timeslot $t$:  

\State 

   \State  \hspace{0.1cm} Disseminate new N-certificates;  \label{Ndis}  \Comment `new' as defined in Section \ref{formal}

      \State  \hspace{0.1cm} Disseminate new stage-2 certificates and new stage-2 M-certificates;  \label{2dis} 

\State 

 \State   \hspace{0.1cm} If $p_i=\mathtt{lead}(\mathtt{v})$ and $\mathtt{ProposeReady}=$ true:

     \State \hspace{0.3cm} $\text{ProposeBlock}$;  \label{esendblock}   \Comment Propose a new block

     \State

      \State   \hspace{0.1cm} If $p_i$ has received $(b,i,c_i,\pi_i)$ that is votable for view $\mathtt{v}$:  \Comment `votable' as defined in Section \ref{formal}
      \State \hspace{0.3cm}  If $1\mathtt{voted}(\mathtt{v})=\text{false}$;  \label{evotecheck}
      \State \hspace{0.5cm} Set $1\mathtt{voted}(\mathtt{v}):=\text{true}$;  Disseminate  a stage-1 vote for $b$ by $p_i$; \label{vote1}   \Comment Disseminate stage-1 vote
      \State \hspace{0.5cm} If $p_i\neq \mathtt{lead}(\mathtt{v})$, disseminate  $(b,i,c_i,\pi_i)$; \label{fdis}

        \State

         \State   \hspace{0.1cm}  If $\mathtt{blocks}^*$ contains a view $\mathtt{v}$ block $b$, $\mathtt{nullified(\mathtt{v})}=$ false and $2\mathtt{voted}(\mathtt{v})=\text{false}$: \label{13}
          \State  \hspace{0.3cm} Set $2\mathtt{voted}(\mathtt{v}):=\text{true}$; Disseminate a stage-2 vote for $b$ by $p_i$;  \label{vote2} \Comment Disseminate stage-2 vote 

        \State

         \State   \hspace{0.1cm}  If $\mathtt{blocks}$ contains a view $\mathtt{v}$ block $b$: 
         
         \State \hspace{0.3cm}  If $\mathtt{nullified(\mathtt{v})}=$ false and $2\mathtt{voted}(\mathtt{v})=\text{false}$: \label{13b}
          \State  \hspace{0.6cm} Set $2\mathtt{voted}(\mathtt{v}):=\text{true}$; Disseminate a stage-2 vote for $b$ by $p_i$;  \label{vote2b} \Comment Disseminate stage-2 vote 
   
         \State  \hspace{0.3cm}  Set  $\mathtt{b}:=b$, $\mathtt{v}:=\mathtt{v}+1$,  $\mathtt{T}:=0$; \label{enewview2} 
         \Comment Update $\mathtt{b}$ and go to next view

               \State
               
               \State  \hspace{0.1cm} If $\mathtt{nullified}(\mathtt{v})=$ false and $\mathtt{TimeoutReady}=\text{true} $:   
                 
                \State \hspace{0.3cm} For all $v\geq \mathtt{v}$ in the same superview as $\mathtt{v}$: 
                 \State \hspace{0.5cm}  Set $\mathtt{nullified}(v):=$ true and disseminate a $\text{nullify}(v)$ message by $p_i$; \label{time-out} 
                 
                 \State \Comment Send nullify$(\mathtt{v})$ upon time-out

        \State 

        \State   \hspace{0.1cm}  If  $\mathtt{S}$ contains an N-certificate for $\mathtt{v}$:
   
        \State  \hspace{0.3cm}  Set $\mathtt{v}:=\mathtt{v}+1$, $\mathtt{T}:=0$; \label{enewview1} 
         \Comment Go to next view

%

       \end{algorithmic}
\end{algorithm}

\section{Carnot 2: the intuition} \label{2int}

In this section, we expand on the intuition behind Carnot~2, which operates under the 
optimal resilience assumption $n \geq 3f+1$ and solves full SMR: 
every correct processor eventually receives every finalised 
transaction. The key difference from Carnot~1 is a recovery 
mechanism that solves SMR rather than Extractable SMR.  Additionally, Carnot~2 incorporates 
a modification to stage-1 voting that streamlines operation 
within superviews. While Carnot~1 was specified with the standard 
partial synchrony model in mind, for Carnot~2 we take care to give 
a specification that also operates cleanly in the pipes 
model~\cite{lewis2025pipes}, which captures the effect of processor 
bandwidth and message sizes on latency and throughput. This 
motivates several of the design choices below and enables the 
throughput analysis of Section~\ref{anal2}.

\subsection{The local variable 
$\mathtt{blocks}$} \label{c2blocks}

In Carnot~1, a processor $p$ progressed to view $v+1$ upon receiving an M-certificate for a view $v$ block $b$ (or an N-certificate for view $v$). The M-certificate sufficed to prove data availability for $b$, even if $p$ has not actually recovered the block. However, since Carnot~2 must solve SMR, data availability no longer suffices. As in Carnot~1, each processor maintains a local variable 
$\mathtt{blocks}$. A block $b$ for view $v$ with parent $b'$ 
(where $b'.\text{view} = v'$) is now added to $\mathtt{blocks}$ by $p_i$ when 
all of the following conditions are met:
\begin{itemize}
\item[(i)] $p_i$ has received a stage-1 certificate for $b$;
\item[(ii)]  $p_i$ has reconstructed the payload of $b$ from certified 
fragments (or recovery fragments, described below);
\item[(iii)]  the parent $b'$ is in $\mathtt{blocks}$;
\item[(iv)]  $\mathtt{S}$ contains an N-certificate for each 
view in the open interval $(v', v)$.\footnote{The requirement for N-certificates in the definition of 
$\mathtt{blocks}$ may appear surprising: in Carnot~1, it was 
sufficient to require an M-certificate or an N-certificate for view 
progression. This requirement is motivated by a modification to 
stage-1 voting described in Section~\ref{c2streamline}, which 
streamlines superview operation by allowing processors to vote 
before seeing stage-1 certificates for ancestor blocks.}
\end{itemize}
 A block 
$b$ is \emph{finalised} when some descendant $b^*$ of $b$ is in 
$\mathtt{blocks}$ and $b^*$ receives a stage-2 certificate.
Since Carnot~2 solves full SMR, we must ensure that every correct 
processor eventually reconstructs the payload of every finalised 
block. This is the role of the recovery mechanism described next.

\subsection{The recovery mechanism} \label{c2recovery}

When a processor $p_i$ adds a block $b$ to $\mathtt{blocks}$, it 
has reconstructed $b$'s payload. However,  other correct processors may 
not yet have enough fragments to do the same. Since Carnot~2 solves 
full SMR, $p_i$ must help them. The mechanism is as follows.

\vspace{0.2cm}
\noindent \textbf{The timer}. Upon adding $b$ to $\mathtt{blocks}$, 
$p_i$ sets a timer (distinct from the timer for nullification). While waiting for the timer to expire, $p_i$ 
continues with other instructions (proposing and voting in 
other views as normal); the timer does not impact the critical 
path during synchrony and when processors act correctly.

When the timer expires, $p_i$ checks which processors have not yet 
sent $p_i$ a stage-2 vote for $b$. Those that have must already 
have reconstructed $b$'s payload (since a correct processor only 
sends a stage-2 vote after adding $b$ to $\mathtt{blocks}$), so 
they need no help. For each remaining processor $p_j$, $p_i$ sends 
$p_j$ (up to) two \emph{certified recovery fragments}: the recovery 
fragment at position $j$ (so that $p_j$ can echo it to all 
processors) and the recovery fragment at position $i$ (giving $p_j$ 
one more fragment toward reconstruction). However,  $p_i$ does not send any 
fragment that has already been exchanged between $p_i$ and $p_j$. We will expand on the form of recovery fragments below: as already explained in Section \ref{lower}, recovery fragments are formed using an  $(n, n-f-1)$-encoding. 

Upon receiving its own recovery fragment together with the 
corresponding stage-1 certificate, processor $p_j$ \emph{echoes} 
that fragment by sending  it to all processors from whom it has not received a stage-2 vote for $b$.  As recovery 
fragments accumulate from multiple sources (direct sends from 
processors whose timers have fired, plus echoed fragments), $p_j$ 
eventually collects enough to reconstruct $b$'s payload. 


\vspace{0.2cm}
\noindent \textbf{Recovery fragments when $k = n-f-1$}. An important 
special case arises when the reconstruction parameter for $b$ is 
$k = n-f-1$. In this case, the primary $(n, n-f-1)$-encoding is 
already conservative enough for recovery: $n-f-1$ fragments being echoed by correct processors suffices 
for reconstruction. Certified recovery fragments are 
therefore \emph{identical} to certified primary fragments in this case. This 
means that $p_i$ need not send a recovery fragment to any processor 
from whom it has already received a certified fragment, nor to any 
processor to whom it has already sent one. In many cases, this 
results in no extra communication at all.

\vspace{0.2cm}
\noindent \textbf{Recovery fragments when $k > n-f-1$: the 
verification problem}. When $k > n-f-1$, the recovery mechanism 
requires re-encoding the payload using the more conservative 
$(n, n-f-1)$-code. This introduces a subtle issue concerning 
verifiability.

In the primary encoding, each fragment includes a validation path against the Merkle root committed to in the block header, 
which is signed by the leader. This allows any processor to verify 
that a fragment is consistent with the proposed block. Recovery 
fragments, however, belong to a \emph{different} erasure coding of 
the same payload, with a different Merkle tree and a different root. 
A processor receiving a recovery fragment has no way to verify it 
against the block header, because the header only contains the 
Merkle root for the primary $(n,k)$-encoding.

This creates an opportunity for attack. A Byzantine processor could 
claim to have reconstructed the payload, produce recovery fragments 
for an arbitrary message, and disseminate them. Correct processors, 
unable to verify these fragments, would echo them. While this does 
not violate consistency, it allows Byzantine processors to cause 
correct processors to disseminate an arbitrary amount of data, which 
is undesirable.

\vspace{0.2cm}
\noindent \textbf{Dual Merkle roots}. The solution is to have the 
leader commit to the recovery encoding at the time of proposal. 
When proposing a block with reconstruction parameter $k > n-f-1$, 
the leader computes \emph{two} erasure encodings of the payload: the 
primary $(n,k)$-encoding and a recovery $(n,n-f-1)$-encoding. The 
block header includes the Merkle roots of both encodings. Since the 
leader already has the full payload, this requires only additional 
computation (two encodings instead of one, potentially carried out in parallel) and adds only a single 
tag to the block header, with no extra communication.

Recovery fragments can now be verified: a processor receiving a 
recovery fragment checks it against the recovery Merkle root in the 
block header, just as it would check a primary fragment against the 
primary Merkle root. Byzantine processors can no longer produce fake 
recovery fragments that will be accepted by correct processors.

\vspace{0.2cm}
\noindent \textbf{Handling a Byzantine leader}. A Byzantine leader 
could commit to a bogus recovery Merkle root, i.e., one that does not 
correspond to the same payload as the primary encoding. To guard 
against this, a correct processor that reconstructs the payload from 
the primary encoding must verify the recovery Merkle root 
\emph{before} adding the block to $\mathtt{blocks}$. It does so by 
re-encoding the payload using the $(n,n-f-1)$-code, computing the 
resulting Merkle root, and checking that it matches the one in the 
block header. If it does not match, the processor knows the leader 
is Byzantine: it does not add the block to $\mathtt{blocks}$, does 
not send a stage-2 vote, and instead nullifies the view.

\vspace{0.2cm}
\noindent \textbf{Encoding compatibility}. With a systematic 
$(n,k)$-erasure code, the payload is divided evenly among the first 
$k$ fragment positions, with the remaining $n-k$ positions holding 
parity data. Similarly, the recovery $(n,n-f-1)$-code divides the 
payload among the first $n-f-1$ positions. Since $k \geq n-f-1$, 
the recovery encoding distributes the same payload across fewer 
systematic positions, meaning each recovery fragment contains more 
data than the corresponding primary fragment. By arranging the 
encoding so that the recovery fragment at each systematic position 
$i \leq n-f-1$ \emph{extends} the primary fragment at position 
$i$ — that is, the primary fragment is a prefix of the recovery 
fragment — a processor sending recovery fragments to $p_j$ at a 
systematic position need only send the \emph{difference}: the 
additional $|C|/(n-f-1) - |C|/k$ data beyond what $p_j$ already 
holds from the primary fragment (together with the new validation path). For 
positions $i > n-f-1$, which hold only parity data in both 
encodings, the full recovery fragment must be sent. This 
optimisation significantly reduces the amount of extra data 
transmitted during recovery, particularly when $k$ is close to 
$n-f-1$.

\subsection{Streamlining superviews} \label{c2streamline}

As noted previously, our aim is to give a specification of Carnot~2 
that allows us to establish Consistency and Liveness in the standard 
model of partial synchrony, but that also operates cleanly in the 
pipes model~\cite{lewis2025pipes}, which captures the effect of 
processor bandwidth and message sizes on latency and throughput. The 
pipes model brings to light certain considerations for optimising 
latency within each superview, which we now describe.

\vspace{0.2cm} 
\noindent \textbf{A timing issue}. For simplicity, suppose (momentarily) that all blocks within a superview are of 
the same size. As in Carnot~1, the leader begins sending each new 
block as soon as it has finished sending the previous one, even if 
other processors have not yet received their fragments. If the 
leader enters the superview at time $t$, then to a first 
approximation\footnote{A more precise calculation is carried out in 
Section~\ref{anal2}.} it finishes sending a first block $b_1$ at 
time $t + s^*$, a second block $b_2$ at time $t + 2s^*$, and so on, 
where $s^*$ is a parameter determined by the block size and 
bandwidth. To minimise latency, we want blocks to be sent as 
frequently as possible  (we want $s^*$ to be small)  since 
transactions then wait less time for inclusion.

However, if $s^*$ is small compared to $\delta$, and if processors cannot disseminate a stage-1 vote for a block 
before receiving a stage-1 certificate for the parent, a cascade of 
$\delta$ delays arises that progressively increases latency within 
the superview. To see this, note that correct processors cannot 
disseminate stage-1 votes for $b_1$ before time $t + s^* + \delta$ 
at the earliest. Since voting for $b_2$ requires a stage-1 
certificate for $b_1$, correct processors cannot vote for $b_2$ 
before $t + s^* + 2\delta$, nor for $b_3$ before 
$t + s^* + 3\delta$, and so on. These are only lower bounds. When processor bandwidths are limited, actual  finalisation requires additional time, including the time for 
processors other than the leader to disseminate their fragments so 
that the payload can be reconstructed, but the lower bounds 
themselves are already growing. As the superview progresses, the 
earliest possible time at which processors can vote on block $b_v$ 
(at least $t + s^* + v\delta$) falls increasingly behind the time 
the leader finishes sending it ($t + vs^*$), causing latency to 
grow with each successive block.

\vspace{0.2cm} 
\noindent \textbf{The solution}.
The approach we take is to allow processors to disseminate a stage-1 vote 
for a block $b$ \emph{before} adding the parent to 
$\mathtt{blocks}$, while still requiring $b$ and all its ancestors 
to be in $\mathtt{blocks}$ before disseminating a stage-2 vote. 
This breaks the cascade: roughly, $b_v$ can be finalised by 
$t + vs^* + 2\delta$  (plus a constant that accounts for fragment dissemination time in the pipes model; see Section~\ref{anal2}), with 
the overhead remaining constant rather than growing with $v$. 
Consistency is not threatened, because of condition~(iv) for 
addition to $\mathtt{blocks}$ specified in 
Section~\ref{c2blocks}: if $b$ for view $v$ is added to 
$\mathtt{blocks}$ and is finalised, then view $v$ does not receive a 
nullification, which prevents any block incompatible with $b$ from 
being added to $\mathtt{blocks}$ in a subsequent view.

\vspace{0.2cm}
\noindent \textbf{Superview-level progression}. Since processors now 
disseminate stage-1 votes before seeing stage-1 certificates for 
ancestor blocks, the notion of sequential progression through views 
within a superview is no longer meaningful: a processor may vote on 
a view $v+1$ block before it has finished processing the view $v$ 
block. Instead, processors progress sequentially through 
\emph{superviews}. Once a processor enters a given superview, it 
sends votes for any view within that superview as soon as it 
receives the requisite messages, without waiting for earlier views 
to complete. The leader disseminates erasure-coded blocks in 
sequential order, and other processors echo the corresponding 
fragments as they arrive.

\vspace{0.2cm}
\noindent \textbf{Timeouts}. This raises the question of how 
timeouts should be orchestrated. In Carnot~1, where processors 
advance through views one at a time, each view has a simple timer 
that fires after a fixed delay. With superview-level progression, we 
instead use a family of conditions, parameterised by the view number 
$v$ within the superview. Roughly, a processor disseminates a 
nullify$(v)$ message, and also nullify$(v')$ messages for all 
subsequent views $v'$ in the same superview, if any of the 
following occurs:
\begin{itemize}
\item[(a)] it does not receive its fragment of a block for view $v$ 
within a time limit that depends on $v$, the bandwidth, the maximum 
block size, and $\Delta$;
\item[(b)] it does not add a view $v$ block to $\mathtt{blocks}$ 
within a corresponding time limit;
\item[(c)] $v$ is not initial and it does not see a view $v-1$ block finalised within a 
corresponding time limit.
\end{itemize}
The time limits increase with $v$ to account for the fact that later 
views within a superview naturally take longer to complete (since the 
leader sends blocks sequentially).   The 
precise values are specified in Section~\ref{2spec}. Note that when 
a processor nullifies view $v$, it also nullifies all subsequent 
views in the same superview. This is to limit the impact of Byzantine leaders.

\section{Carnot 2: The formal specification} \label{2spec} 
   The pseudocode uses a number of message types, local
variables, predicates, functions and procedures. Those different from Section \ref{formal} are described below. 

\vspace{0.2cm}
\noindent  \textbf{The function} $\mathtt{lead}(w)$.  For superview $w$, the leader is $\mathtt{lead}(w):= p_{j+1}$, where $j \equiv w \text{ mod }n$. 

\vspace{0.2cm}
\noindent \textbf{Blocks}. In Section \ref{formal}, a block other than the genesis block\footnote{The genesis block can remain unchanged or, to give it the same structure as other blocks, can now be of the form $(0,\tau(\lambda,n),\tau(\lambda,n),\lambda)$.},  with associated payload $C$, was a tuple $b=(v,\tau(C,k), h)$ signed by $\mathtt{lead}(v)$ (for some  $k\in [n-f-1,n-1]$). Now blocks are tuples of the form  $b=(v,\tau(C,k), \tau(C,n-f-1), h)$ signed by $\mathtt{lead}(w)$, where view $v$ belongs to superview $w$.  To be correctly formed, it must hold that if  $\tau(C,k)= (\beta, k,r)$ for some $\beta\in \mathbb{N}$ and some hash value $r$, then $\tau(C,n-f-1)=(\beta,n-f-1,r')$ for some hash value $r'$.  As in Section \ref{formal}, we set $b.\text{tag}:=\tau(C,k)$, but now we also set $b.\text{rtag}:= \tau(C,n-f-1)$.  If $(c_i,\pi_i)$ is a certified fragment of $\tau(C,k)$ at $i$, we also say that the tuple $(b,i,c_i,\pi_i)$ is a \emph{certified fragment of} $b$ at $i$.

\vspace{0.2cm}
\noindent \textbf{Stage-$d$ votes}. For $d\in \{ 1,2 \}$, a  \emph{stage-$d$ vote} by $p_i\in \Pi$ for the block $b$ is a message of the form $(\text{vote},b,d,i,\rho_i)$, where $\rho_i$  is a signature share from $p_i$ on the message $(\text{vote},b,d)$, using an $(n-f)$-of-$n$ threshold signature scheme.

\vspace{0.2cm}
\noindent \textbf{Stage-$d$ notarisations and certificates}.  A  \emph{stage-$d$ notarisation} for the block $b$ is a set of at least $n-f$ stage-$d$ votes for $b$, each by a different processor in $\Pi$.  A \emph{stage-$d$ certificate} for the block $b$ is the message $(d\text{Cert},b,\rho)$, where $\rho$ is an $(n-f)$-of-$n$ threshold certificate on the message $(\text{vote},b,d)$. 

\vspace{0.2cm}
\noindent \textbf{Nullifications and N-certificates}. For $v\in \mathbb{N}_{\geq 1}$, a nullify$(v)$ message by $p_i$ is of the form $(\text{nullify},v,i,\rho_i)$, where $\rho_i$ is a signature share from $p_i$ on the message $(\text{nullify},v)$, using an $(n-f)$-of-$n$ threshold signature scheme. A \emph{nullification} for view $v$ is a set of $(n-f)$ nullify$(v)$ messages, each by a different processor in $\Pi$.  An \emph{N-certificate} for view $v$ is a message $(\text{N-cert}, v,\rho)$, where $\rho$ is a $(n-f)$-of-$n$ threshold certificate on the message $(\text{nullify,}v)$.

\vspace{0.2cm}
\noindent \textbf{The local variable} N-$\mathtt{certificates}$. Initially empty, this local variable for $p_i$ is automatically updated by $p_i$ (without explicit instructions in the pseudocode) to contain all  N-certificates received by $p_i$. If $p_i$ receives a nullification, then it automatically forms the associated certificate, and adds it to N-$\mathtt{certificates}$.


\vspace{0.2cm} 
\noindent \textbf{The local value $\mathtt{blocks}$}.  For any set of messages $M$, $\mathtt{blocks}(M)$ is set to be the smallest set of blocks containing 
$b_{\text{gen}}$ and all blocks $b=(v,(\beta,k,r),(\beta,n-f-1,r'), h)$ such that (i)-(iv) below are satisfied: 
\begin{itemize}
\item[(i)] $M$ contains a stage-1 certificate for $b$;
\item[(ii)] $M$ contains certified fragments (or \emph{recovery fragments}, see below) of $b$ sufficient to recover the payload and verify that it decodes correctly, and that encoding the payload with reconstruction parameters $k$ and $n-f-1$ produces tags $(\beta,k,r)$ and $(\beta,n-f-1,r')$ respectively;
\item[(iii)] There exists $b'\in \mathtt{blocks}(M)$ with $H(b')=b.\text{par}$ and $b'.\text{view}<v$, i.e., the parent $b'$ of $b$ is in $\mathtt{blocks}(M)$. Let $v':=b'.\text{view}$;
\item[(iv)]  $M$ contains an N-certificate for each 
view in the open interval $(v', v)$.
\end{itemize} 
We also write  $\mathtt{blocks}$ to denote $\mathtt{blocks}(\mathtt{S})$. 

\vspace{0.2cm}
\noindent \textbf{The local variable} $\mathtt{w}$. Initially set to 1, this variable specifies the present superview of a processor.

\vspace{0.2cm}
\noindent \textbf{The local timers} $\mathtt{T}(b)$. When $p_i$ adds $b$ to $\mathtt{blocks}$, it sets a timer to expire\footnote{While waiting for the timer expire, $p_i$ continues with other instructions. It is important to note that the timer does not impact the critical path under good conditions, and is only used to determine whether and when $p_i$ should send extra fragments of $b$.} in time $s$. Here, $s$ is a timing parameter. When the reconstruction parameter $k$ for $b$ is greater than $n-f-1$, $s$ should be roughly the expected time  (during synchrony and when processors act correctly) until receiving stage-2 votes from all (or most) processors.  When $k=n-f-1$, $s$ can be chosen more conservatively to be the minimum of the latter value and the expected time until $p_i$ receives fragments of $b$ from all (or most) processors. 

\vspace{0.2cm}
\noindent \textbf{Certified recovery fragments}.  If the reconstruction parameter $k$ for $b$ is $n-f-1$, then a \emph{certified recovery fragment} of $b$ at $i$ is identical to a certified fragment of $b$ at $i$. If $k>n-f-1$, then suppose $b=(v,\tau(C,k), \tau(C,n-f-1),h)$ signed by $\mathtt{lead}(w)$ (where $w$ is the superview containing $v$).  Suppose $ \text{Encode}(C,n-f-1):= (\tau(C,n-f-1), \{ (c_i,\pi_i ) \}_{i\in [n]} )$. Then the \emph{certified recovery fragment}\footnote{For proving Consistency and Liveness in the standard model of partial synchrony, it does not matter whether we use the `encoding compatibility' optimisation of Section \ref{c2recovery}. We consider implementations using this optimisation  in Section \ref{exper}.} of $b$ at $i$ is the tuple $(\text{rec}, b,i,c_i,\pi_i)$.  Suppose $\tau(C,n-f-1)= (\beta, n-f-1,r')$ for some $\beta\in \mathbb{N}$ and some hash value $r'$. For the recovery fragment to be correctly formed, it must hold that: 
\begin{itemize} 
\item $c_i$ is of the correct length (given $n$ and $f$) to be a fragment of a message of length $\beta$, and;
\item $\pi_i$ is a validation path from $r'$ to $c_i$ at position $i$. 
\end{itemize}

\vspace{0.2cm}
\noindent \textbf{The procedure} ProposeBlock. Let $\mathtt{w}$ be as locally defined for $p_i$. If $p_i=\mathtt{lead}(\mathtt{w})$, this procedure is executed by $p_i$ to produce and send out a new block. To execute the procedure, $p_i$ proceeds as follows:
\begin{itemize}
\item If $p_i$ has not previously proposed any blocks for superview $\mathtt{w}$, set $v$ to be the first view of superview $\mathtt{w}$  and set $b^\ast:=\mathtt{b}$. The block $b^\ast$ will be the parent of the new proposed block.  Otherwise, let $v-1$ be the greatest view in superview $\mathtt{w}$ such that $p_i$ has already proposed a view $v-1$ block $b'$, and set $b^\ast:=b'$. 
\item  Form a payload $C$, containing all transactions  received  but not included in the payloads of ancestors of $b^\ast$.
\item Set $k:=g(\mathtt{S},v)$ and calculate Encode$(C,k)=(\tau(C,k), \{ (c_j,\pi_j ) \}_{j\in [n]} )$ and  Encode$(C,n-f-1)=(\tau(C,n-f-1), \{ (c'_j,\pi'_j ) \}_{j\in [n]} )$.
\item Set $b:=(v,\tau(C,k),  \tau(C,n-f-1), H(b^\ast))$ signed by $p_i$ and, for each $j\in [n] \setminus \{ i \}$, send $(b,j,c_j,\pi_j)$ to $p_j$, i.e., send the certified fragment of $b$ at  $j$ to $p_j$ (we say `$p_i$ proposes $b$'). 
\end{itemize}

\vspace{0.2cm}
\noindent \textbf{The local predicate $\mathtt{ProposeReady}$}.  This predicate is essentially the same as in Section \ref{formal}, but must now be redefined to incorporate the fact that processors no longer maintain a local variable $\mathtt{v}$. If $\mathtt{w}$ is as locally defined for $p_i$ and $p_i=\mathtt{lead}(\mathtt{w})$, there are two cases: 

\vspace{0.1cm} 
\noindent \emph{Case 1}: $p_i$ has not previously proposed any blocks for superview $\mathtt{w}$. In this case, $\mathtt{ProposeReady}=\text{true}$.

\vspace{0.1cm} 
\noindent \emph{Case 2}: Otherwise.  Then $\mathtt{ProposeReady}$ is true if and only if $p_i$ has `finished' sending the previous block  and has not already proposed blocks for all views in the superview. As in Section \ref{formal},  such considerations make sense in the context of the Pipes model \cite{lewis2025pipes}, where processors have limited bandwidth, and will be discussed further in Section \ref{anal2}. Our formal assumptions on the $\mathtt{ProposeReady}$ predicate when analysing the protocol in the standard model of partial synchrony (with no limit on the size of messages that can be sent in a timeslot) will be made explicit when we define the $\mathtt{TimeoutReady}$ predicates below.

\vspace{0.2cm}
\noindent \textbf{The local predicates $\mathtt{TimeoutReady}(v)$ for $v\in \mathtt{N}_{\geq 1}$}.  Recall that $s$ is a timing parameter, used to determine when $p_i$ should send extra fragments. We also consider a second parameter $s^*$, which should be thought of as depending on bandwidth and maximum block size. If $v$ is the $j^{\text{th}}$ view in its superview, then  $\mathtt{TimeoutReady}(v)$ is set to true if any  of the three following conditions apply, and is otherwise false: 
\begin{enumerate} 
\item[(a)] $\mathtt{T}=3\Delta+s +js^*$, and  $1\mathtt{voted}(v)=2\mathtt{voted}(v)=\text{false}$, or;                 
\item[(b)] $\mathtt{T}=4\Delta+2s+js^*$, and  $2\mathtt{voted}(v)=\text{false}$:
\item[(c)]  $\mathtt{T}=5\Delta+2s+js^*$ and there does not exist any view $v$ block in $\mathtt{blocks}$ for which there also exists a stage-2 certificate in $\mathtt{S}$. 
\end{enumerate} 
We note that, for $j=1$, the threshold in (a) above is $3\Delta+s +s^*$. However, the \emph{difference} between the corresponding thresholds for (general) $j$ and $j+1$ is only $s^*$. Except at throughputs close to the full processor bandwidth, and under realistic network conditions, $s^*$ should also be thought of as a small fraction $\delta$ (we will consider these matters in more detail in Section \ref{anal2}). \emph{So the difference in timeouts for successive views is small}. Similar considerations also apply to the timeout thresholds in (b) and (c). 

To explain the timeout thresholds above, we note that our recovery mechanism will ensure all correct processors enter each superview within time $2\Delta+s$ of each other after GST. 
When proving Liveness in the standard model of partial synchrony, we will formally assume only that if the leader $p_i$ of superview $w$ is correct and enters the superview at $t\geq \text{GST}$, then it proposes the $j^{\text{th}}$ block for the superview by $t+js^*$. This means that if $p_j$ enters at $t'\geq \text{GST}$, the leader $p_i$ enters by $t'+2\Delta +s$. Processor $p_j$ then receives a certified fragment of $p_i$'s proposal for the first view of the superview by $t'+3\Delta+s+s^*$, and so on. 


\vspace{0.2cm}
\noindent \textbf{The function $\mathcal{F}$ for SMR}.  Given any set of messages $M$, let $b_1,\dots,b_m$ be the longest sequence of blocks in $\mathtt{blocks}(M)$ such that, for each $i<m$, $b_{i}$ is the parent of $b_{i+1}$, and, for each $i\leq m$, there exists $j\geq  i$ such that $M$ contains a stage-2 certificate for $b_j$. Then we define $\mathcal{F}(M)$ to be the sequence of transactions formed by concatenating the payloads of $b_1,\dots,b_m$. If there does not exist a unique longest sequence $b_1,\dots,b_m$ as specified above for $M$, then we set  $\mathcal{F}(M)$ to be the empty sequence.

\vspace{0.2cm}
 The pseudocode is shown in Algorithms \ref{alg3} and \ref{alg4}. Appendix \ref{anal3} analyses Carnot~2 in the standard model of partial synchrony.

 \begin{algorithm}
\caption{: Carnot 2, the instructions for $p_i$, MAIN LOOP.}
 \label{alg3}
\begin{algorithmic}[1]


\State  At every timeslot $t$:  

\State 

   \State  \hspace{0.1cm} Disseminate new N-certificates and stage-1 certificates; \label{2Ndis}   \Comment `new' as defined in Section \ref{formal}

\State 

 \State   \hspace{0.1cm} If $p_i=\mathtt{lead}(\mathtt{w})$ and $\mathtt{ProposeReady}=$ true:

     \State \hspace{0.3cm} $\text{ProposeBlock}$;  \label{2sendblock}   \Comment Propose a new block

     \State
       
      \State   \hspace{0.1cm} For each view $v$ in superview $\mathtt{w}$ such that  $1\mathtt{voted}(v)=\text{false}$:  \label{2votecheck}
      \State   \hspace{0.3cm} If $p_i$ has received a certified fragment $(b,i,c_i,\pi_i)$ of a view $v$ block $b$:  
           \State \hspace{0.5cm} Set $1\mathtt{voted}(v):=\text{true}$;  Disseminate  a stage-1 vote for $b$ by $p_i$; \label{2vote1}   \Comment Disseminate stage-1 vote
      \State \hspace{0.5cm} If $p_i\neq \mathtt{lead}(\mathtt{w})$, disseminate  $(b,i,c_i,\pi_i)$; \label{2fdis}

        \State
        
                  \State \hspace{0.1cm} For each new $b\in \mathtt{blocks}$: 
          
          \State \hspace{0.3cm} Trigger the timer $\mathtt{T}(b)$ to expire in time $s$;  \Comment $s$ a parameter

%

         \State  \hspace{0.3cm}   Let $v=b.\text{view}$. If $v >\mathtt{b}.\text{view}$, set  $\mathtt{b}:=b$;
         
           \State  \hspace{0.3cm}  If  $\mathtt{nullified}(v)=2\mathtt{voted}(v)=$ false: 
           
              \State  \hspace{0.5cm}   Disseminate a stage-2 vote for $b$ by $p_i$;  Set $2\mathtt{voted}(v):=$ true;  \label{2vote2} \Comment Disseminate stage-2 vote

               \State
               
               \State  \hspace{0.1cm} If there exists a least view $v$ in superview $\mathtt{w}$ with $\mathtt{TimeoutReady}(v)=\text{true} $:

                \State \hspace{0.3cm} For all $v'\geq v$ in superview $\mathtt{w}$ with $\mathtt{nullified}(v')=2\mathtt{voted}(v')=$ false: 
                 \State \hspace{0.5cm}  Set $\mathtt{nullified}(v'):=$ true and disseminate a $\text{nullify}(v')$ message by $p_i$; \label{2time-out} 
                 
                 \State \Comment Nullify upon time-out

        \State 

        \State   \hspace{0.1cm}  If  N-$\mathtt{certificates} \ \cup\  \mathtt{blocks}$ contains either an N-certificate  or a block for each  \label{newwclause} 
        
         \State   \hspace{0.1cm}  view in superview $\mathtt{w}$: 
   
        \State  \hspace{0.3cm}  Set $\mathtt{w}:=\mathtt{w}+1$, $\mathtt{T}:=0$; \label{2newsuperview} 
         \Comment Go to next superview

       \end{algorithmic}
\end{algorithm}

 \begin{algorithm}
\caption{: Carnot 2, the instructions for $p_i$ to enumerate extra fragments}
 \label{alg4}
\begin{algorithmic}[1]

\State  At every timeslot $t$:

              \State 
              
              \State \hspace{0.1cm} For each $b$ such that $\mathtt{T}(b)$ expires at $t$: \label{f1a}

\State \hspace{0.3cm} Let $w$ be the superview containing $b.\text{view}$;

\State \hspace{0.3cm} For each $j \in [n]$ with $p_j \neq \mathtt{lead}(w)$ and from whom $p_i$ has not received a stage-2 vote for $b$:

\State \hspace{0.5cm} Send $p_j$ the certified recovery fragment of $b$ at $j$, unless already sent to or received from $p_j$;

\State \hspace{0.5cm} If $p_i \neq \mathtt{lead}(w)$, send $p_j$ the certified recovery fragment of $b$ at $i$, unless already sent to or 

\State \hspace{0.5cm}  received from $p_j$; \label{f1b}

\State 

\State  \hspace{0.1cm}  For any block $b$ s.t.\ all of the following apply:  \label{sendowns}
\State  \hspace{0.1cm}  (i) $p_i$ has received $F$ which is a certified recovery fragment of $b$ at $i$;
\State  \hspace{0.1cm}  (ii) $p_i$  has not previously sent $F$ to any processor;
\State  \hspace{0.1cm}  (iii) For the superview $w$ containing  $b.\text{view}$, $p_i\neq \mathtt{lead}(w) $, and; 
\State  \hspace{0.1cm} (iv) $\mathtt{S}$ contains a stage-1 certificate for $b$:  

\State \hspace{0.3cm} For each $j \in [n]$ with $p_j \neq \mathtt{lead}(w)$ and from whom $p_i$ has not received a stage-2 vote for $b$:

\State \hspace{0.3cm}  Send $F$ to $p_j$ if not previously received from $p_j$;      \label{f2b} 

\end{algorithmic}
\end{algorithm}

\section{Analysis in the Pipes Model} \label{anal2}

Standard models in distributed computing treat delivery times as 
independent of message size: a message sent at time $t$ in the 
synchronous setting is guaranteed to arrive by $t+\Delta$, regardless 
of how much data it carries. This makes such models ill-suited to 
reasoning about throughput or real-world latency, since they abstract 
away the bandwidth constraints that govern practical performance. The 
pipes model of Lewis-Pye, Nayak, and Shrestha~\cite{lewis2025pipes} 
addresses this by introducing explicit per-processor bandwidth, 
allowing one to express latency as a function of bandwidth, the 
incoming transaction rate, the message delay, $n$, and 
protocol-specific parameters, and to identify the \emph{latency 
bottleneck}, i.e., the maximum incoming transaction rate the protocol can 
sustain without unbounded latency. In this section, we analyse Carnot~2 in the pipes model. 

\vspace{0.2cm}
\noindent \textbf{Why Carnot~2?} As anticipated in Section~\ref{2int}, 
the specification of Carnot~2 was crafted with the pipes model in 
mind. The way blocks are erasure-coded sequentially within a 
superview, the way fragments are echoed, and the form of the timeout 
conditions of Section~\ref{2spec} are all chosen so that throughput 
can be analysed directly in terms of the bandwidth $S$, the superview 
length $x$, and the erasure coding parameter $k$, without requiring 
further protocol-level adjustments. Carnot~1, by contrast, was 
specified with the standard model of partial synchrony in mind, and 
adapting it for a pipes-model analysis would require non-trivial 
changes to its view-level timing; we leave this to future work.

\vspace{0.2cm}
\noindent \textbf{Structure of the section.} 
Section~\ref{pipes_model} recalls the relevant aspects of the pipes 
model. Section~\ref{pipes_approach} gives a high-level overview of 
how the model applies to Carnot~2, identifying the key quantities the 
analysis must track. Section~\ref{pipes_analysis} carries out the 
analysis itself. Section \ref{pipes_compare}  compares with Dispersed Simplex~\cite{shoup2023sing}, which is the natural comparison point for leader-based protocols using erasure codes. 
Finally, Section \ref{DAG_compare}  compares with DAG-based protocols, taking Sailfish~\cite{shrestha2025sailfish}  as a concrete comparison point.

%
%
\subsection{The Model} \label{pipes_model}

We briefly recall the relevant aspects of the Pipes model~\cite{lewis2025pipes}.

\vspace{0.2cm}
\noindent \textbf{Processors and communication.}
We consider a set of $n$ processors $\Pi=\{p_1, \ldots, p_n\}$, each maintaining a direct channel to every other processor.  Time is divided into discrete timeslots, and we assume a uniform message delay of $\delta$ slots for sending a single information parcel between any two processors.

\vspace{0.2cm}
\noindent \textbf{Bandwidth.}
Each processor has a bandwidth of $S$ bits per timeslot. Formally, each processor maintains a single \emph{upload buffer} and a single \emph{download buffer}. When a processor sends a message at time slot $t$, the corresponding bits (each addressed to a recipient) are added to its upload buffer. At the end of timeslot $t$, $S$ bits (or all bits in the buffer if there are less than $S$) are removed from the upload buffer in FIFO order and appear on the download buffers of their intended recipients at $t + \delta$. At the start of each timeslot,  $S$ bits (or all bits in the download buffer if there are less than $S$) are removed from the download buffer and are received by the processor. So, upload and download rates are both bounded by $S$, and are assumed equal for all processors.

\vspace{0.2cm}
\noindent \textbf{Clients and transaction arrival.}
Each logical processor is modelled as a \emph{consensus processor} paired with a \emph{client processor}, connected by an infinite-bandwidth zero-delay channel. External transactions arrive at client processors and are forwarded to the corresponding consensus processor. In the single-sender (leader-based) setting of Carnot, the next leader's client processor receives transactions at a rate of $D$ bits per timeslot. When analysing DAG-based protocols, such as Sailfish~\cite{shrestha2025sailfish}, we suppose the $D$ transaction bits per timeslot arriving at the network as a whole are divided evenly between the $n$ processors, so that each receives $D/n$ transaction bits per timeslot. 

\vspace{0.2cm}
\noindent \textbf{Latency.}
The \emph{latency} of a transaction is the time from when it arrives at a correct client processor to when it is finalised by all correct processors. A protocol has \emph{bounded latency} at arrival rate $D$ if latency remains finite over an arbitrarily long execution; the \emph{latency bottleneck} is the supremum of arrival rates for which bounded latency holds.

\vspace{0.2cm}
\noindent \textbf{Simplifying assumptions.}
Following~\cite{lewis2025pipes}, we ignore computational costs (signature verification, erasure encoding and decoding), assume all processors have the same bandwidth $S$, and assume uniform  delay $\delta$ between all pairs. These assumptions mean the analysis reflects an idealised lower bound on achievable latency, abstracting away implementation-level overheads. Following~\cite{lewis2025pipes}, we also restrict analysis to the \emph{good case} when the network is synchronous and processors are correct. There are two reasons for this. First, real-world deployments of SMR protocols experience substantial faulty or asynchronous behaviour only rarely; the common case is one in which the network behaves well and most processors follow the protocol. Protocol designers therefore typically optimise for good-case performance while ensuring safety and liveness are maintained under adversarial conditions. Second, analysis in the pipes model is already considerably more involved than classical round-based reasoning; focusing on the good case allows us to isolate the throughput and latency results of greatest practical interest while keeping the analysis tractable.

\subsection{Overview} \label{pipes_approach} 

To minimise latency, it is generally desirable for the leader to 
include all pending transactions\footnote{For fixed timeout values, 
i.e., fixed values of $s$ and $s^*$ in the case of Carnot~2, this 
holds so long as block sizes do not grow to the point where timeouts 
are triggered.} when forming a block proposal: any transaction not 
included must wait for a later block, and this wait can add 
substantially to its latency. Optimal block sizes are therefore a 
function of the time between successive block proposals, the incoming 
transaction rate $D$, the bandwidth $S$, the number of processors $n$, and other parameters, such as vote sizes. For Carnot~2, this 
means that block sizes depend on the position of the corresponding 
view within its superview. 

\vspace{0.2cm}  
The first block of each superview is the largest. Before proposing, 
the new leader must wait to add the final block of the previous 
superview to its local variable $\mathtt{blocks}$, and transactions 
accumulate at the client processor during this wait. For the second block, the 
leader need only wait until it has finished sending the first. Crucially, it can begin this dissemination before other processors have 
finished receiving fragments of the first block. The second block 
therefore need contain only those transactions that arrived while 
the first was being sent, and is correspondingly smaller. The same 
principle applies to the third block, which can be smaller still: 
because the second block was smaller than the first, less time was required to send 
it, so fewer transactions accumulated in the meantime. 

\vspace{0.2cm}  
Block sizes thus decrease towards an equilibrium value within each 
superview. To capture this behaviour, the analysis in 
Section~\ref{pipes_analysis} computes \emph{two} latencies. First, we 
compute transaction latency at equilibrium within a superview. With 
this in hand, we then compute latency for transactions included in 
the first block of a superview, which upper-bounds latency across 
all transactions. The relative importance of these two quantities 
depends on the superview length $x$: when $x$ is large, the 
equilibrium latency is the relevant value for most transactions.

\subsection{Analysis} \label{pipes_analysis} 

To carry out the analysis, we first describe some assumptions regarding messages sizes. 

\vspace{0.2cm} 
\noindent \textbf{Message sizes}. In Section \ref{2spec}, we specified a block (excluding payload) as a tuple $b=(v,\tau(C,k), \tau(C,n-f-1), h)$ signed by the leader of the superview.  We let $\lambda$ be a parameter used to specify the length of hashes, and approximate the length of this tuple as three hash values and one signature. Since commonly used signature schemes  have signatures of length approximately two hash values, we then approximate the size of $b$ as $5\lambda$.  In Section \ref{2spec}, votes were of the form $(\text{vote},b,d,i,\rho_i)$, where $\rho_i$  is a signature share from $p_i$. While this message includes $b$, it suffices to include a hash of $b$, so we approximate votes as being of size $3\lambda$. We suppose stage-1 certificates are of size $2\lambda$. If a block has payload of size $B$, and if erasure coding is carried out with data expansion rate $d$, we suppose the size of a fragment is $(Bd/n)+ (5+\log n)\lambda$: the $\log n$ term corresponds to the length of the validation path.  Since we restrict to the good case, we need not size nullify messages, N-certificates, or recovery fragments, which are not transmitted in correct, synchronous executions.

\vspace{0.2cm} 
\noindent \textbf{Assumptions regarding message priorities}. When the leader $p_i$  sends a fragment to $p_j$, we suppose that $p_j$ begins forwarding that fragment to others from the moment it begins receipt. As in \cite{lewis2025pipes}, we suppose message sending from any processor to all others is `balanced', so that if $p_i$ sends a message to all others, they will receive it at the same time (subject to constraints on their download buffers). When a processor is simultaneously instructed to send multiple messages, we suppose they prioritise stage-2 votes, then stage-1 votes, then stage-1 certificates, and then finally blocks and their corresponding fragments. 

\vspace{0.2cm} 
\noindent \textbf{Calculating latency at equilibrium}. First, we must calculate block sizes. Suppose the leader $p_i$ starts sending certified fragments of $b$ at $t$. While sending these fragments, $p_i$ must also disseminate one stage-2 vote and one stage-1 certificate per view, and can then immediately disseminate a stage-1 vote for $b$. This takes time:\footnote{For the sake of simplicity, we suppose that sending a message to all others requires sending $n$ copies of the message, rather than $n-1$.} 

\[ T:=\frac{Bd}{S} +\frac{(5+\log n)\lambda n}{S}+\frac{3\lambda n}{S}+\frac{2\lambda n}{S}+\frac{3\lambda n}{S}= \frac{Bd+ (13+\log n)n\lambda}{S}.\] 

\noindent During this interval of length  $T$, $DT$ transaction bits arrive at the client processor. So, at equilibrium, $B=DT$. This means: 
\[ \frac{D(Bd +(13+\log n)n\lambda)}{S}=B.\]
\noindent Therefore: 
\[ B=\frac{D(13+\log n)n\lambda}{S-Dd}. \]  
Feeding the value of $B$ back into our expression for $T$ gives: 
\[T= (13+\log n)n\lambda \left( \frac{Dd}{S(S-Dd)}+\frac{1}{S} \right)= \frac{ (13+\log n)n\lambda}{S-Dd}.\]

\noindent Recall that the leader  $p_i$ starts sending fragments at $t$, and that other processors start forwarding them immediately upon receipt. Since sending is balanced,  download buffers are never a bottleneck in this analysis.  This means other processors have received and finished sending their fragments by $t+T+\delta$. At this point, they must disseminate a stage-1 vote, which takes time $3\lambda n/S +\delta$ to be removed from their upload buffers and be received by others. At this time, they will disseminate stage-2 votes, which take a further time  $3\lambda n/S +\delta$ to be received by others. The block is therefore finalised by all correct processors by: 
\[ t+ T+\frac{6n\lambda}{S} +3\delta= t+ \frac{ (13+\log n)n\lambda}{S-Dd}+ \frac{6n\lambda}{S} +3\delta.\] 
To complete the latency calculation, we must include the time (at most) $T$ that each transaction waits to be included in a block, giving total latency: 
\begin{equation} \label{eqlat}  \frac{ (13+\log n)2n\lambda}{S-Dd}+ \frac{6n\lambda}{S} +3\delta.\end{equation} 

\vspace{0.2cm} 
\noindent \textbf{Interpreting the formula}. We note that for an incoming transaction rate well below $S/d$, and for reasonable values of $n$ and $S$, the first two terms above will be small, and the $3\delta$ term will dominate. The first term only becomes significant as $D$ approaches $S/d$ and the denominator tends to 0. This means the latency tends to infinity as $D$ approaches $S/d$: the latter value  is the latency bottleneck.

\vspace{0.2cm}
\noindent \textbf{Block times at equilibrium}. The value $T$ computed above is also the \emph{block time} at equilibrium: the leader begins sending the fragments of each block as soon as it has finished sending those of the previous block (together with the accompanying per-view votes and certificates), so successive proposals are separated by time
\[ T= \frac{(13+\log n)n\lambda}{S-Dd}.\]
Two features of this expression are worth noting. First, it does not contain $\delta$: block times are determined by bandwidth and per-view overhead alone, and may be far smaller than the message delay. Second, the payload contributes only through the term $Dd$ in the denominator. For incoming transaction rates well below $S/d$, the block time therefore approaches the floor $(13+\log n)n\lambda/S$, which is just the time the leader requires to send the per-view votes, certificates, and validation paths. For $\lambda=256$ and $S=10^9$ bits per second, this floor is approximately $0.5$ms when $n=100$, and approximately $2.2$ms when $n=400$. (At such block times, the binding constraints are of a different nature: the equilibrium block size is $B=DT$, so blocks are meaningfully non-empty only if transactions arrive at a rate of at least one per block time.) Block times of this order also put the superview length $x$ in perspective: a leader controlling $x=100$ consecutive views holds office for $xT$, i.e., a fraction of a second, as discussed in Section~\ref{intro}.

\vspace{0.2cm}
\noindent \textbf{Calculating first block latency}. To calculate latency for transactions included in the first block of a superview, suppose the previous leader starts sending the final block of the previous superview at $t_1$. Let $T$ be as above, i.e.,  $T := \frac{(13+\log n)n\lambda}{S-Dd}$, and note that:  
\[ T  = \frac{DTd}{S} + \frac{(13+\log n)n\lambda}{S}, \]
\noindent which rearranges to 
\begin{equation} \label{clev}  \frac{Dd}{S} T = T - \frac{(13+\log n)n\lambda}{S}. \end{equation} 

\noindent From the analysis above, it follows that all correct processors will add the final block of the previous superview to their local value $\mathtt{blocks}$ by: 
\[ t_2:= t_1+T+ \frac{3n\lambda}{S} + 2\delta,\] 
\noindent  which means that at $t_2$ the leader of the next superview will form a block of size: 
\[ B:= D(T+ \frac{3n\lambda}{S} + 2\delta). \] 

\noindent All processors will receive and finish forwarding on their fragments of the new block by: 
 \[ t_3:= t_2+ \frac{Bd}{S}+  \frac{(5+\log n)n\lambda}{S}+\delta.\] 
 \noindent  From Equation (\ref{clev}), it follows that: 
\[ \frac{Bd}{S} = \frac{Dd}{S} T + \frac{Dd}{S}\left( \frac{3n\lambda}{S} + 2\delta \right) 
= T - \frac{(13+\log n)n\lambda}{S} + \frac{Dd}{S}\left( \frac{3n\lambda}{S} + 2\delta \right). \]
Adding $\frac{(5+\log n)n\lambda}{S}$ collapses the $(13+\log n)-(5+\log n)=8$, giving 
\[ \frac{Bd}{S} + \frac{(5+\log n)n\lambda}{S} 
= T - \frac{8n\lambda}{S} + \frac{Dd}{S}\left( \frac{3n\lambda}{S} + 2\delta \right).  \]
So: 
\[ t_3-t_2= \frac{Bd}{S} + \frac{(5+\log n)n\lambda}{S} +\delta 
= \frac{(13+\log n)n\lambda}{S-Dd} - \frac{8n\lambda}{S} + \frac{3Ddn\lambda}{S^2} + \frac{2Dd\delta}{S} +\delta. \]

\noindent For reasonable parameter values $3Ddn\lambda/S^2$ will be small. Dropping this term gives: 

\[ t_3-t_2=  \frac{(13+\log n)n\lambda}{S-Dd} - \frac{8n\lambda}{S} + \frac{2Dd\delta}{S} +\delta. \]
\noindent The block will then be finalised by: 
\[ t_3+ \frac{6n\lambda}{S} + 2\delta,\] 
giving total latency: 
\begin{equation} \frac{2(13 + \log n)n\lambda}{S-Dd} +\frac{n\lambda}{S} + \frac{2Dd\delta}{S}+5\delta. \end{equation}

\subsection{Comparison with DispersedSimplex} \label{pipes_compare} 

The natural single-sender comparison point for Carnot~2 is 
DispersedSimplex~\cite{shoup2023sing}, which also uses erasure coding for 
block dissemination and achieves $n\geq 3f+1$ resilience. The 
comparison, however, involves a few subtleties, since 
DispersedSimplex admits several variants. We discuss each in turn. 

\vspace{0.2cm} 
\noindent \textbf{The pipes-paper analysis.} An analysis of 
DispersedSimplex in the pipes model already appears in 
\cite{lewis2025pipes}. However, that analysis assumes a single leader that remains in place
\emph{throughout the duration of the execution}, allowing for fixed block
sizes, and considers the basic variant with data expansion rate~$3$.

\vspace{0.2cm} 
\noindent \textbf{Achieving data expansion rate 1.5.} For a like-for-like 
comparison with Carnot~2, one should consider the version of 
DispersedSimplex that uses rotating multi-view leaders (`stable leaders' in the terminology of \cite{shoup2023sing}) and that  achieves data
expansion rate $1.5$, since this is the lowest rate
DispersedSimplex can attain. This variant is also discussed in 
\cite{shoup2023sing}, but here further complexities arise. 

\vspace{0.2cm} 
\noindent  \emph{The formally specified version}. In the version that 
\cite{shoup2023sing} formally specifies with pseudocode, the leader must enter
each view before proposing for that view. This produces significantly greater latency than the equilibrium figure for Carnot~2 (Equation (\ref{eqlat})) at the same data expansion rate.   
Carrying out the 
analysis in the pipes model (see Appendix~\ref{pipesDSanal}) gives latency: 
\begin{equation} 
 \left( \frac{(6+\log n)2n\lambda}{S} + 4\delta \right) \cdot \frac{1}{1 - 1.5D/S} + \frac{3n\lambda}{S} + \delta.
 \end{equation} 

\vspace{0.2cm} 
\noindent  \emph{The pipelined variant}. \cite{shoup2023sing} also informally 
discusses the possibility of pipelining: having the leader send blocks 
without first entering each view. The pipelined variant is not 
formally specified, however, and raises subtleties of the 
kind addressed in Section \ref{c2streamline} of this paper. In particular, if processors must wait 
to receive a stage-1 certificate for the parent before voting on a 
block, then the pipelined variant is subject to exactly the timing 
issue discussed in Section\ref{c2streamline}: the propagation of $\delta$ 
delays causes latency to grow without bound within the superview. 
For sufficiently large blocks (i.e., for sufficiently high incoming 
transaction rates), this is not an issue, since the bandwidth-limited 
send time for each block exceeds $\delta$ and the $\delta$ delays 
are absorbed; in this regime the pipelined variant matches\footnote{This holds modulo minor differences caused by different approaches to message formatting.} the 
latency of Carnot~2 \emph{at data expansion rate~$1.5$}. At lower 
throughputs, however, the propagation of $\delta$ delays causes 
latency to grow within the superview, and the issue cannot be 
resolved without a mechanism analogous to that of Carnot~2.

\vspace{0.2cm} 
\noindent \textbf{Summary.} In the regime where the most aggressive 
informal variant of (multi-view leader) DispersedSimplex is well-behaved (sufficiently
high throughput), it matches Carnot~2 \emph{at data expansion 
rate~$1.5$}.  In contrast, Carnot~2 (i)~avoids the timing pathology 
at lower throughputs without requiring a high-throughput assumption, 
and (ii)~can push the data expansion rate below~$1.5$, towards~$1$, 
giving a strictly higher throughput bottleneck $S/d$. The formally 
specified version of DispersedSimplex with data expansion rate~$1.5$ 
has higher latency than Carnot~2 at the same data expansion rate.

\vspace{0.2cm} 
\noindent Figure \ref{Firstpic} illustrates these comparisons with concrete parameter values. In the figure, we 
 set $S=10^9$ bits per second, $\lambda=256$, $n=400$, and $\delta=0.1$. The displayed numbers for DispersedSimplex, are for the formally specified multi-view-leader version of the protocol, with data expansion rate 1.5.

\begin{figure}
    \centering
    \includegraphics[width=0.8\linewidth]{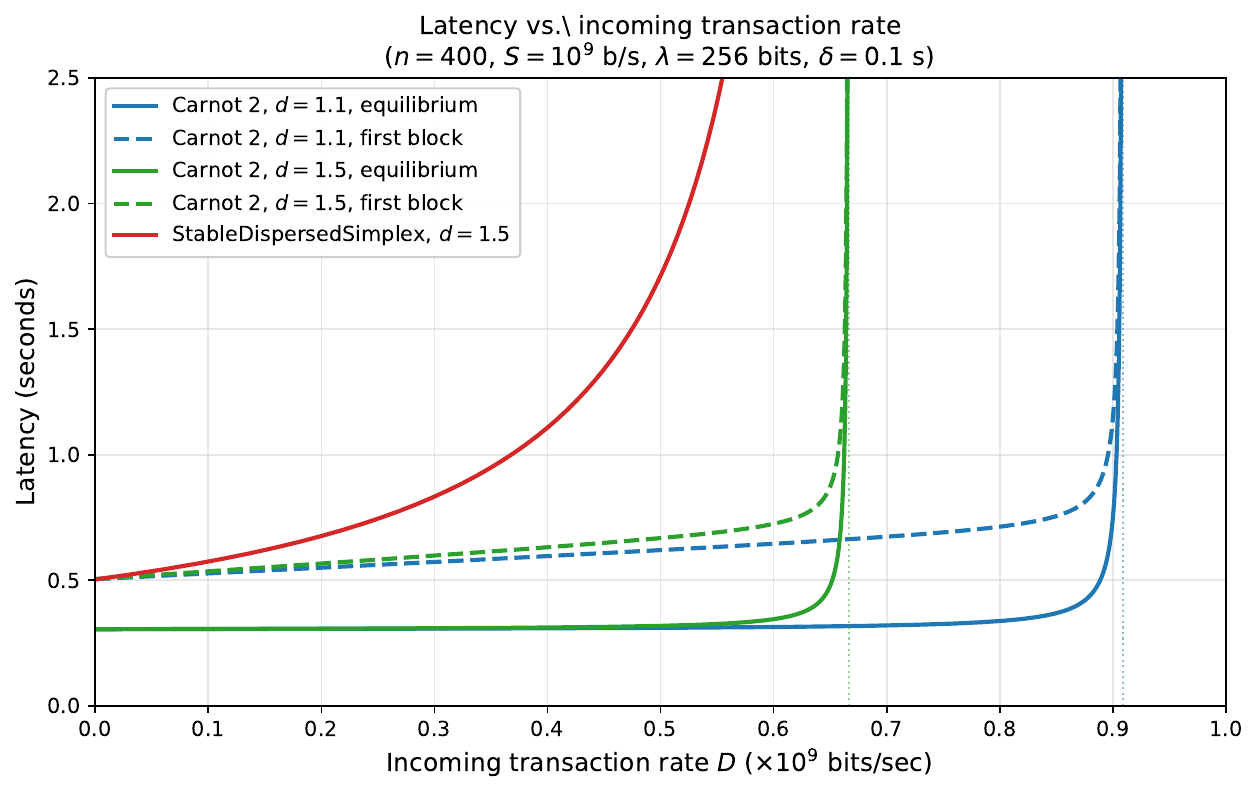}
    \caption{Latency for Carnot~2 and DispersedSimplex.}
    \label{Firstpic}
\end{figure}

\subsection{Comparison with Sailfish} \label{DAG_compare} 

We use Sailfish as a concrete point of comparison for DAG-based 
protocols, since it has competitive latency among such protocols. 
An analysis of Sailfish in the pipes model was already carried out 
in \cite{lewis2025pipes}, and we import their analysis here. As 
noted in \cite{lewis2025pipes}, although the Sailfish paper suggests 
using Reliable Broadcast as the underlying mechanism for block 
propagation, doing so defeats the purpose of using a DAG-based 
protocol: the requirement that processors echo each other's 
proposals reduces the latency bottleneck from $O(S)$ to $O(S/n)$. 
We therefore follow \cite{lewis2025pipes} in supposing that block 
propagation is carried out via a form of Consistent Broadcast: 
processors send their blocks to all others, who then disseminate a 
signed message acknowledging receipt. A quorum of such signatures 
then certifies data availability for the block. In this form, the 
protocol solves Extractable SMR rather than SMR. The latency 
formula derived in \cite{lewis2025pipes} is: 
\begin{equation} \label{LatSailCon2} \left( \frac{9\lambda n^2}{S}+9\delta \right) \left(1+ \frac{8}{9((S/D)-1)} \right). 
\end{equation}
We note that this figure is \emph{not} the time to finalise leader blocks: latency incorporates the time for transactions to be included in a block, the time for that block to be pointed to by a leader block, and then the time to finalise that leader block (in the good case). 

\vspace{0.2cm} 
\noindent Figure \ref{2pic} illustrates  latencies for Carnot~2 and Sailfish, for the same parameter values as Figure \ref{Firstpic}. 

\vspace{0.2cm} 
\noindent  The basic trade-off between Carnot~2 and Sailfish is that, while Carnot~2 has significantly lower latency over a wide range of parameter values and solves SMR rather than just Extractable SMR, Sailfish has a latency bottleneck that is greater by a factor $d$ (the data expansion rate for Carnot~2). 

\begin{figure}
    \centering
    \includegraphics[width=0.8\linewidth]{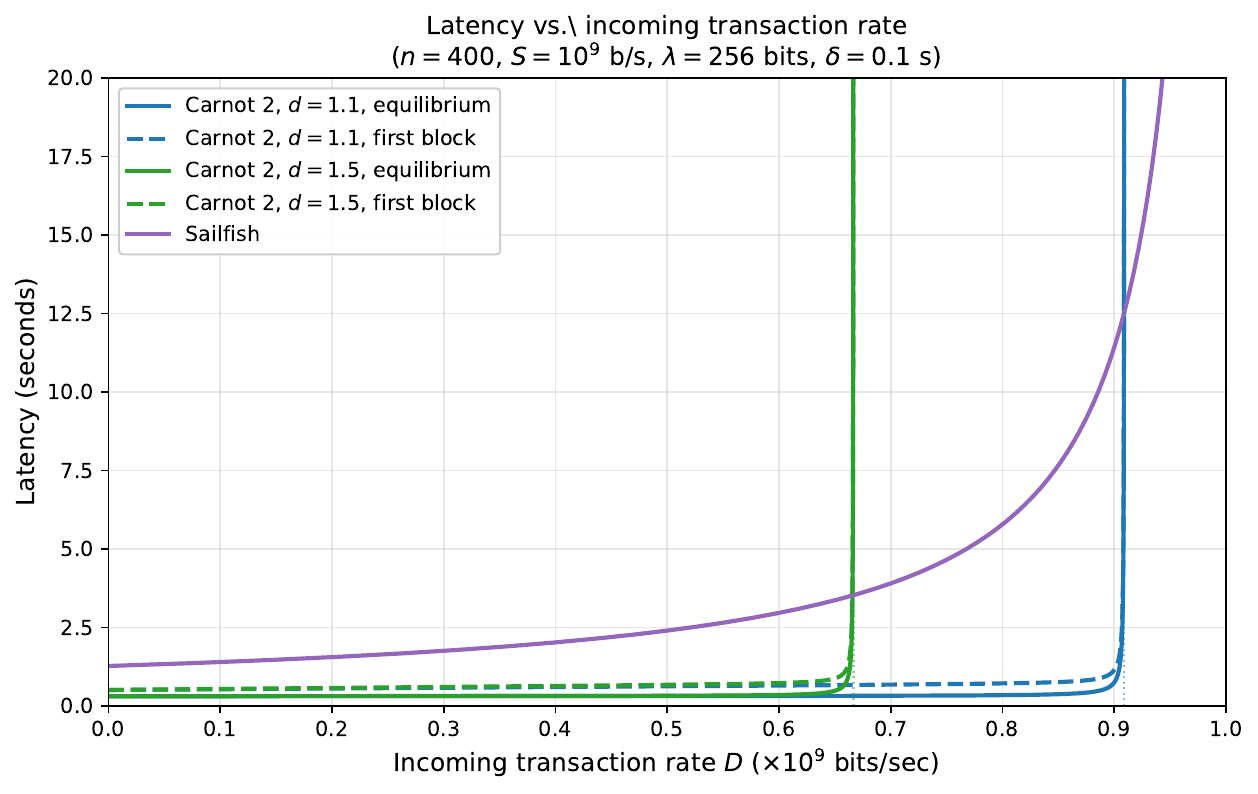}
    \caption{Latency for Carnot~2 and Sailfish: parameters are the same as for Figure \ref{Firstpic}.}
    \label{2pic}
\end{figure}

\section{The data expansion bound for 2-round finality} \label{imposs} 

In this section, we establish bounds on coding efficiency for protocols with 2-round finality. We first explain the connection between SMR protocols with 2-round finality and Byzantine Broadcast. We then define 2-round Byzantine Broadcast formally, state our theorem, and prove it using an indistinguishability argument.

\vspace{0.2cm}
\noindent \textbf{From 2-round finality to Byzantine Broadcast}. An SMR protocol has \emph{2-round finality} if, when the leader is correct, the network is synchronous, and sufficiently many processors are correct, a block proposed by the leader is finalised after two rounds of message exchange.  Protocols with 2-round finality, such as E-Minimmit~\cite{chou2025minimmit}, Hydrangea~\cite{shrestha2025hydrangea}, and Kudzu~\cite{shoup2025kudzu}, assume $n\geq 5f+1$ ($n\geq 5f-1$ is necessary and sufficient~\cite{kuznetsov2021revisiting}).

We formalise the relevant notion of Byzantine Broadcast below, and then explain how any SMR protocol with 2-round finality yields a solution.

\subsection{2-round Byzantine Broadcast} \label{2rBB}

We define 2-round Byzantine Broadcast (2-round BB) in the partially synchronous model with $\Delta = 1$. Let $f^*\leq f$ be a parameter. 

\begin{itemize} 
\item We consider a set $\Pi$ of $n$ processors, of which at most $f$ display Byzantine faults. 
\item One processor is designated the \emph{broadcaster}. All processors are told the identity of the broadcaster. 
\item The broadcaster is given an input in some set $V$. The set $V$ is known to all processors. 
\item The protocol must satisfy the following conditions: 
\begin{itemize} 
\item \textbf{Termination}. All correct processors must eventually output a value in $V$. 

\item \textbf{Agreement}. No two correct processors output different values.

\item \textbf{Validity}. If $\text{GST}=0$, the broadcaster is correct with input $v$, and at most $f^*$ processors are faulty, then all correct processors output $v$ by timeslot 2. 
\end{itemize} 
\end{itemize} 

\vspace{0.2cm}
\noindent \textbf{From SMR with 2-round finality to 2-round BB}. Any SMR protocol with 2-round finality yields a protocol for 2-round BB as follows. Given an input $v$, the broadcaster acts as the first leader and proposes a block containing $v$. Processors then execute the full SMR protocol and output the first finalised payload. Validity follows from the protocol's liveness guarantee: if $\text{GST}=0$, the broadcaster is correct, and at most $f^*$ processors are faulty (for these protocols, $f^*=f$), the block is finalised within two rounds. Agreement follows from Consistency. Termination follows from the protocol's general liveness mechanism---processors eventually finalise some block, even if this requires additional rounds. Our lower bound therefore applies to any SMR protocol with 2-round finality.

\begin{theorem} \label{2rBBT}
Suppose $f\geq 2$, $n\geq 3f+1$ and that $\mathcal{P}$ is a polynomial-time protocol solving 2-round BB. Let $P\subseteq \Pi$ be any set of size $n-2f-f^*$ that does not contain the broadcaster. Then, in any execution of $\mathcal{P}$ in which the broadcaster is correct, the broadcaster's input can be computed in polynomial time from the messages sent by the broadcaster to processors in $P$ at timeslot 0.
\end{theorem}

\noindent The requirement $f\geq 2$ is used in the proof to ensure that a certain set of processors is non-empty for the indistinguishability argument (specifically, the set $P_4$ in Section~\ref{2rproof}).

\vspace{0.2cm}
Our proof establishes a natural way to compute the broadcaster's input from the messages sent to processors in $P$ at timeslot~0. Let $P$ be as in the statement of the theorem, and consider any execution $E$ in which $\text{GST}=0$ and:
\begin{itemize} 
\item The broadcaster has input $v$ and correctly sends messages to processors in $P$ at timeslot 0, but crashes immediately after doing so, without sending messages to any processor in $\Pi\setminus P$. 
\item A set of $f-1$ processors in $\Pi\setminus P$, other than the broadcaster, crash after correctly sending messages at  timeslot 0. 
\item All other processors act correctly. 
\end{itemize} 
We will show that all correct processors output $v$ in $E$. The broadcaster's input can therefore be computed from the messages sent by the broadcaster to processors in $P$ at timeslot 0, simply by simulating $E$ without knowledge of $v$. 

\vspace{0.2cm} 
\noindent \textbf{Interpreting the theorem}. Protocols with 2-round finality, such as E-Minimmit, Hydrangea, and Kudzu, assume $n\geq 5f+1$ and solve 2-round BB with $f^*=f$. When $n=5f+1$, Theorem \ref{2rBBT} establishes that the broadcaster's input can be recovered from the messages sent to any set of $n-2f-f^* = 2f+1$ processors, i.e., just over $2/5$ths of the total. If the broadcaster uses erasure coding and sends a single fragment to each processor, recovery from $2f+1$ fragments requires a reconstruction parameter $k\leq 2f+1$, giving a data expansion rate of at least $n/k \geq (5f+1)/(2f+1)$, which approaches $2.5$ as $f$ grows. This bound is exactly tight: E-Minimmit and Kudzu use erasure codes with reconstruction parameter $k=2f+1$, achieving data expansion rates of $(5f+1)/(2f+1)\approx 2.5$.

\subsection{The proof of Theorem \ref{2rBBT}} \label{2rproof}

Suppose $f\geq 2$, $n\geq 3f+1$, and that $\mathcal{P}$ is a polynomial-time protocol solving 2-round BB. Let $P\subseteq \Pi$ be any set of size $n-2f-f^*$ that does not contain the broadcaster, and let $v$ be an arbitrary input to the broadcaster. Writing $L$ for the broadcaster, we partition $\Pi$ into five disjoint sets $\{L\}, P_1, P_2, P_3, P_4$, where $P_1=P$, $|P_2|=f$, $|P_3|=f^*$, and $|P_4|=f-1$. Note that $P_4$ is non-empty because $f\geq 2$. We consider three executions, $E_1$, $E_2$, and $E_3$, as described below. 

\vspace{0.1cm}
\noindent \emph{Execution $E_1$} is specified as follows:
\begin{itemize}
\item $\text{GST}=0$.
\item The broadcaster correctly sends messages to processors in $P_1$ at timeslot 0, but crashes immediately after doing so, without sending messages to any processor in $\Pi\setminus P_1$.
\item All processors in $P_4$ crash after correctly sending messages at timeslot 0.
\item All other processors act correctly.
\end{itemize}

\noindent Note that the number of faulty processors in $E_1$ is $1+|P_4|=f$. By Termination, all correct processors must output by some timeslot in $E_1$. Choose $t^*>2$ greater than this timeslot.

\vspace{0.1cm}
\noindent \emph{Execution $E_2$} is specified as follows:
\begin{itemize}
\item $\text{GST}=t^*$.
\item The broadcaster is correct with input $v$. Processors in $P_2$ are faulty, while all other processors are correct.
\item Prior to $t^*$, processors in $P_2$ send the same messages to processors in $P_1 \cup P_2 \cup P_3$ as in $E_1$, i.e., they act as if they are correct but did not receive any message from the broadcaster.
\item At timeslot 0, processors in $P_2$ send the same messages to $\{L\}\cup P_4$ as in $E_1$ (note that timeslot-0 messages from correct processors are the same in all three executions, since they are sent before any messages are received). At timeslot 1, they send messages to $\{L\}\cup P_4$ as if correct and having received the broadcaster's timeslot-0 message at timeslot 1. They do not send messages to $\{L\}\cup P_4$ at later timeslots.
\item Messages from the broadcaster to processors in $P_3$ sent at timeslots in $[0,t^*)$, and messages from the broadcaster to processors in $P_1$ sent at timeslots in $(0,t^*)$, are delivered at $\text{GST}=t^*$.
\item Messages from processors in $P_4$ to processors in $P_1 \cup P_3$ sent at timeslots in $(0,t^*)$ are delivered at $\text{GST}=t^*$.
\item All other messages are delivered at the next timeslot.
\end{itemize}

\vspace{0.1cm} 
\noindent \emph{Execution $E_3$} is specified as follows: 
\begin{itemize} 
\item $\text{GST}=0$. 
\item All processors are correct, except those in $P_3$. 
\item Processors in $P_3$ act as if they are correct but did not receive any message from the broadcaster.
\end{itemize} 

\vspace{0.1cm} 
\noindent \textbf{Analysis}. In $E_3$, the broadcaster is correct with input $v$, $\text{GST}=0$, and at most $f^*$ processors are faulty, so by Validity all correct processors output $v$ by timeslot 2. 

We claim that $E_2$ and $E_3$ are indistinguishable for processors in $P_4$ through timeslot 2. First note that timeslot-0 messages from correct processors are the same in all three executions, and that processors in $P_2$ send the same timeslot-0 messages to $\{L\}\cup P_4$ in $E_2$ as in $E_3$ by construction. It follows that, in $E_2$ and $E_3$, the broadcaster, and processors in $P_1$ and $P_4$, all receive the same messages at timeslot 1 (in both executions, they receive the timeslot-0 messages of all processors, including the broadcaster's), and so have the same views through timeslot 1. Since the broadcaster is correct in both executions, it therefore sends the same messages at timeslot 1 in both, and these are received by processors in $P_4$ at timeslot 2. The same reasoning applies to the timeslot-1 messages of processors in $P_1$ and of processors in $P_4$ themselves. Processors in $P_3$ also send the same messages through timeslot 1 in both executions: in $E_2$ they are correct but have not received any message from the broadcaster (such messages being delayed until $t^*$), while in $E_3$ they received the broadcaster's message but act as if they did not. Finally, processors in $P_2$ send the same messages to $P_4$ at timeslots 0 and 1 in both executions, by construction. So processors in $P_4$ have the same view through timeslot 2 in $E_2$ as in $E_3$. Since $P_4$ is non-empty (as $f\geq 2$), and since processors in $P_4$ output $v$ by timeslot 2 in $E_3$, processors in $P_4$ output $v$ by timeslot 2 in $E_2$. By Agreement, all correct processors in $E_2$ output $v$; in particular, processors in $P_1$ output $v$ in $E_2$.

We claim that $E_1$ and $E_2$ are indistinguishable for processors in $P_1$ prior to timeslot $t^*$. In both executions, processors in $P_1$ receive the same timeslot-0 message from the broadcaster at timeslot 1, and receive no further messages from the broadcaster prior to $t^*$: in $E_1$ because the broadcaster crashes after timeslot 0, and in $E_2$ because the broadcaster's messages to $P_1$ sent at timeslots after 0 are delayed until $t^*$. In $E_1$, processors in $P_4$ crash after timeslot 0, so no processor in $P_1 \cup P_3$ receives any message from $P_4$ after timeslot 1. In $E_2$, messages from $P_4$ to $P_1 \cup P_3$ sent at timeslots after 0 are delayed until $t^*$, so again no processor in $P_1 \cup P_3$ receives any message from $P_4$ after timeslot 1 and before $t^*$. In $E_2$ and  $E_1$,  processors in $P_2$ send the same messages to processors in $P_1 \cup P_2 \cup P_3$ prior to $t^*$ by construction, while processors in $P_3$ are correct and have the same view prior to $t^*$ in both executions (in both cases, they receive no message from the broadcaster and no messages from $P_4$ after timeslot 1). So processors in $P_1$ have the same view in both executions prior to $t^*$. Since processors in $P_1$ output before $t^*$ in $E_1$, they output the same value in both executions. Therefore, processors in $P_1$ output $v$ in $E_1$. By Agreement, all correct processors output $v$ in $E_1$.

The broadcaster's input can therefore be computed from the messages sent by the broadcaster to processors in $P$ at timeslot 0, simply by simulating $E_1$.

\andy{Potentially comment on extension to Extractable SMR.}

\section{Experiments}  \label{exper} 
\andy{To be added.} 

\section{Related work}  \label{rw} 

\noindent \textbf{Classical Byzantine Consensus}. The study of
protocols for reaching consensus in the presence of Byzantine
faults was introduced by Lamport, Shostak, and
Pease~\cite{lamport1982byzantine}, with a treatment of SMR given later by
Schneider~\cite{schneider1993replication}. Dwork, Lynch and
Stockmeyer~\cite{DLS88} showed that $n\geq 3f+1$ is optimal for
partial synchrony. Standard protocols using this assumption, such
as PBFT~\cite{castro1999practical} and
Tendermint~\cite{buchman2016tendermint,buchman2018latest}, satisfy
3-round finality. As shown by~\cite{abraham2021good}, this is
optimal. More recently, Chan and
Pass~\cite{chan2023simplex} introduced Simplex, a particularly
simple protocol for partial synchrony with rotating leaders and
two rounds of voting per view. Simplex achieves 3-round finality
with $n \geq 3f+1$ and has a clean structure that lends itself
well to extension. Our protocols build directly on Simplex,
extending it with erasure coding, multi-view leaders, and the
mechanisms for pushing the data expansion rate towards $1$
described in this paper.

\vspace{0.2cm}
\noindent \textbf{Erasure coding in distributed protocols.}
The use of erasure codes to reduce communication costs in
distributed protocols has a long history. Cachin and
Tessaro~\cite{cachin2005asynchronous} introduced
\emph{asynchronous verifiable information dispersal (AVID)},
combining erasure codes with Merkle trees to allow a broadcaster to
distribute data among $n$ servers so that the data can be
recovered from any sufficiently large subset of fragments, with
each fragment verifiable against a commitment in the block
header. Our use of certified fragments follows this approach
directly. As detailed below, a number of other recent SMR protocols also make use of this approach. 
The Pipes framework of Lewis-Pye, Nayak and
Shrestha~\cite{lewis2025pipes} analyses the throughput
implications of erasure coding in leader-based protocols, showing
that the data expansion rate directly governs the throughput
bottleneck. Our work builds on this line of research, establishing
that the $2.5$ rate is optimal for 2-round finality and showing
how to circumvent this bound with 3-round finality. Shoup~\cite{shoup2023sing} has previously shown how to achieve a data expansion rate of 1.5 with 3-round finality.

\vspace{0.2cm}
\noindent \textbf{Protocols with 2-round finality.}
A long line of
work~\cite{brasileiro2001consensus,friedman2005simple,guerraoui2007refined,kursawe2002optimistic,martin2006fast,song2008bosco}
considers protocols with a `fast path', allowing quick
finalisation under favourable conditions. FaB~\cite{martin2006fast}
introduced a parameterised model with $n \geq 3f + 2p + 1$
processors, achieving 2-round finality when at most $p$ processors
are Byzantine. Kuznetsov et
al.~\cite{kuznetsov2021revisiting} later showed that the optimal
bound is $n \geq 3f + 2p - 1$ (see
also~\cite{abraham2021good}). Recent protocols in this family
include Minimmit and
E-Minimmit~\cite{chou2025minimmit},
Alpenglow~\cite{alpen},
Kudzu~\cite{shoup2025kudzu},
Hydrangea~\cite{shrestha2025hydrangea},
Banyan~\cite{vonlanthen2024banyan}, and
ChonkyBFT~\cite{francca2025chonkybft}. Among these,
E-Minimmit, Kudzu and Alpenglow incorporate erasure coding, achieving data
expansion rates of approximately $2.5$. Our impossibility result
(Section~\ref{imposs}) shows that this rate cannot be improved
for protocols with 2-round finality.

\vspace{0.2cm}
\noindent \textbf{Optimistic proposals and multi-view leaders.}
In rotating-leader protocols, a standard bottleneck is that the
leader of view $v+1$ must wait for a certificate from view $v$
before proposing.  Moonshot~\cite{doidge2024moonshot} introduced the
concept of \emph{optimistic proposals}, in which the leader of view
$v+1$ proposes a block as soon as it receives the proposal for
view $v$, without waiting for the corresponding certificate. 
Hydrangea++~\cite{shresthahydrangeaplusplus} applies the same technique to
Hydrangea. In this paper, we use  \emph{multi-view leaders} (where a single leader proposes
across multiple consecutive views) to allow the next leader to
begin proposing before the current superview has concluded,
eliminating the inter-superview gap. A number of previous protocols (e.g., PBFT~\cite{castro1999practical} and DispersedSimplex~\cite{shoup2023sing}) have incorporated multi-view leaders, often under the name `stable leaders'.\footnote{We also note that a previous version of this paper (still available on the arXiv) shows how optimistic proposals can be integrated with Carnot~1.}

\vspace{0.2cm}
\noindent \textbf{DAG-based protocols.} It has often been observed
that low communication complexity does not necessarily translate
into high throughput in practice. Narwhal and
Tusk~\cite{danezis2022narwhal} demonstrated that building on a
DAG-based dissemination layer, in which all processors share
responsibility for transmitting transactions, can yield
significantly higher throughput than leader-based approaches.
Narwhal further improves throughput at the systems level by
allowing each processor to employ multiple \emph{worker nodes}
for data transmission, effectively trading additional hardware
resources (CPUs and network bandwidth) for performance. While
increased bandwidth can be incorporated into our model, the use
of multiple CPUs per processor is not. Because data and protocol
messages are transmitted by many processors concurrently,
DAG-based protocols typically incur higher communication
complexity than their leader-based counterparts. Beyond serving
as a dissemination layer, a number of
works~\cite{gkagol2018aleph,keidar2021all,spiegelman2024shoal,shrestha2025sailfish,keidar2022cordial,babel2023mysticeti}
use the DAG structure itself to reach consensus. Early DAG-based
protocols suffered from high round-latency, and reducing this to
match leader-based protocols has been an active area of research.

\vspace{0.2cm}
\noindent \textbf{Small block times.} A recent line of work aims to
minimise the inter-proposal time, pushing it below the message delay
bound $\Delta$. Gatling~\cite{scaffino2026gatling} runs $K$ staggered
parallel instances of a black-box atomic broadcast protocol and
interleaves their outputs via a deterministic merge rule, achieving
inter-proposal times of $T_{\text{ipt}}/K$ for a component protocol
with inter-proposal time $T_{\text{ipt}}$.
Cadence~\cite{elsheimy2026cadence} similarly decouples the block
interval from the network delay by finalising each slot in its own
independent consensus instance, with multiple concurrent proposers
per slot. Both protocols require synchronised clocks, and both, in
their basic forms, sacrifice \emph{predictable validity}: when
proposals are issued at sub-$\Delta$ intervals by rotating proposers,
a proposer cannot have seen the immediately preceding blocks at
proposal time, and so cannot validate the transactions it includes
against an up-to-date state of the log. The authors
of~\cite{scaffino2026gatling} present two variants of Gatling that
retain predictable validity, but both operate under what they term
slowly rotating leader schedules: a single leader issues multiple
consecutive proposals, with consecutive leader windows separated by
gaps of length at least $\Delta$ (either left empty, or filled with a
second tier of state-independent transactions). This is precisely the
multi-view-leader regime in which Carnot operates, and in this regime
the parallel composition offers no advantage over direct pipelining
within a superview. Moreover, because a single leader now
disseminates every (prime) block in its window, and because the
composition treats the component protocol as a closed box (without
erasure coding), the leader's outgoing bandwidth is reinstated as the
throughput bottleneck---the very problem that multi-proposer designs
set out to address, and which Carnot addresses with erasure coding.
As established in Section~\ref{anal2}, Carnot~2 achieves equilibrium
block times of $(13+\log n)n\lambda/(S-Dd)$---of the order of a
millisecond for realistic parameter values---while retaining
predictable validity (a Carnot leader always builds on blocks whose
payloads it holds), requiring no synchronised clocks, and allowing
data expansion rates approaching~$1$. Cadence, for comparison,
disseminates proposals as erasure-coded chunks of which $f+1$ suffice
for reconstruction, giving a data expansion rate of approximately
$3$; its multiple concurrent proposers target censorship resistance
properties that are orthogonal to the concerns of this paper.

\section{Discussion}  \label{disc} 

We have shown that protocols with 2-round finality cannot
achieve a data expansion rate below $2.5$, and that protocols
with 3-round finality can push the data expansion rate
arbitrarily close to $1$. Two protocols realising this were
presented: Carnot~1, which assumes $n \geq 4f+1$, solves Extractable SMR,  and avoids any
extra fragment dissemination, and Carnot~2, which solves SMR  under
the optimal assumption $n \geq 3f+1$ at the cost of additional
fragment dissemination when Byzantine processors interfere.

Beyond bandwidth efficiency, we have argued that multi-view leaders
give small block times essentially for free: at equilibrium, the
interval between successive proposals is determined by bandwidth and
per-view overhead alone (Section~\ref{anal2}), and is of the order of
a millisecond for realistic parameter values. In light of recent work
targeting sub-$\Delta$ inter-proposal
times~\cite{scaffino2026gatling,elsheimy2026cadence}, we regard this as
a further argument for the multi-view-leader design: it achieves
comparable block times without synchronised clocks and without
sacrificing predictable validity, and---since the wall-clock length
of a leader's tenure at such block times is a fraction of a
second---at little cost to the concerns, such as censorship, that
motivate rapid leader rotation.

Several questions remain open. First, our impossibility result
applies to protocols with 2-round finality (one round of voting).
It would be interesting to establish tight bounds on the
achievable data expansion rate for protocols with 3-round finality
in the \emph{worst case}. Carnot~1 achieves approximately $1.33$ and
Carnot~2 achieves approximately $1.5$, but we do not know whether
these rates are optimal for their respective resilience
assumptions.

Second, while the Pipes model analysis of
Section~\ref{anal2}  provides formulas relating 
erasure coding parameters and maximum throughput, an interesting
practical question is how to set these parameters dynamically in
realistic settings where network conditions and fault rates are in
flux. More broadly, one might wish to move between protocols with
2-round finality (which require fewer rounds of communication but a higher data
expansion rate) and 3-round finality (which can achieve lower data
expansion rates and thus higher throughput) depending on current
demand. Understanding how to smoothly transition between these
regimes is an appealing direction for future work.

\section*{Acknowledgements}

We thank Sunghyeon Jo for a note improving Theorem~\ref{2rBBT}: an earlier version of this paper established the theorem for sets $P$ of size $n-2f-f^*+1$, and the observation that the broadcaster may be kept correct in the pivotal execution, allowing the tightening to $n-2f-f^*$, is due to him.


\appendix 

\section{Carnot 1: analysis in the standard model} \label{anal1} 
Throughout this section, we assume $n\geq 4f+1$, and we consider the standard model of partial synchrony, as described in Section \ref{setup}. We say a block $b$ \emph{receives a stage-$2$ certificate} if some processor receives a stage-$2$ certificate for $b$, and we also use similar terminology for stage-2 M-certificates and N-certificates.  We say $b$ \emph{receives a stage-1 notarisation} if at least $n-\frac{3}{2}f$ processors send stage-1 votes for $b$.

\subsection{Consistency} \label{econsec}
The proof of consistency follows a similar structure to that for Simplex~\cite{chan2023simplex}.

\begin{lemma}[One vote per view] \label{singlevote}
Correct processors disseminate at most one stage-1 vote and at most one stage-2 vote in each view.
\end{lemma}
\begin{proof}
For stage-1 votes, note that a correct processor $p_i$ only disseminates a stage-1 vote for a view $v$ block in line~\ref{vote1}, and only if $1\mathtt{voted}(v) = \text{false}$. Upon doing so, $p_i$ sets $1\mathtt{voted}(v) := \text{true}$. For stage-2 votes, observe that $p_i$ only disseminates a stage-2 vote for a view $v$ block in lines~\ref{vote2} and \ref{vote2b}, and only if $2\mathtt{voted}(v) = \text{false}$. Upon doing so, $p_i$ sets $2\mathtt{voted}(v) := \text{true}$.
\end{proof}

\begin{lemma}[No stage-2 vote and nullify] \label{stage2ornull}
No correct processor both disseminates a stage-2 vote for a view $v$ block and sends a \emph{nullify}$(v)$ message.
\end{lemma}
\begin{proof}
A correct processor $p_i$ only disseminates a nullify$(v)$ message while in view $v' \leq v$ (line~\ref{time-out}) if $\mathtt{TimeoutReady}=\text{true}$. From the definition of the latter predicate, $\mathtt{TimeoutReady}=\text{true}$ implies  $2\mathtt{voted}(v')=\text{false}$. Since $p_i$ only ever sets $2\mathtt{voted}(v):=\text{true}$ while in view $v$, this also implies that $2\mathtt{voted}(v)=\text{false}$. Processor $p_i$ sets $\mathtt{nullified}(v) := \text{true}$ upon  disseminating the nullify$(v)$ message. Furthermore, $p_i$ only disseminates a stage-2 vote for a view $v$ block while in view $v$ (lines~\ref{vote2} and \ref{vote2b}), and only if $\mathtt{nullified}(v) = \text{false}$. 
\end{proof}

\begin{lemma}[Unique stage-1 notarisation per view] \label{1certlem}
If block $b$ receives a stage-1 notarisation, then no block $b' \neq b$ with $b'.\text{view} = b.\text{view}$ receives a stage-1 notarisation.
\end{lemma}
\begin{proof}
This follows from Lemma~\ref{singlevote} by the standard quorum intersection argument. Suppose, towards a contradiction, that $b$ receives a stage-1 notarisation $Q$ and $b' \neq b$ with $b'.\text{view} = b.\text{view}$ receives a stage-1 notarisation $Q'$. Let $P$ and $P'$ be the sets of processors contributing to $Q$ and $Q'$, respectively. Then $|P \cap P'| \geq (n-\frac{3}{2}f) + (n-\frac{3}{2}f) - n = n - 3f \geq f + 1$. Thus $P \cap P'$ contains at least one correct processor, contradicting Lemma~\ref{singlevote}.
\end{proof}

\begin{lemma}[Stage-2 certificate precludes N-certificate] \label{2certlem}
If block $b$ receives a stage-2 certificate, then view $v := b.\text{view}$ does not receive an N-certificate.
\end{lemma}
\begin{proof}
This follows from Lemma~\ref{stage2ornull} by a quorum intersection argument. Towards a contradiction, suppose $b$ receives a stage-2 certificate $Q$ and that  $v := b.\text{view}$ receives an N-certificate $Q'$.  Let $P$ and $P'$ be the sets of processors contributing to $Q$ and $Q'$, respectively. Then $|P \cap P'| \geq (n-f) + (2f+1) - n =f + 1$. Thus $P \cap P'$ contains at least one correct processor, contradicting Lemma~\ref{stage2ornull}.
\end{proof}


\begin{lemma}[Stage-2 certificates imply stage-1 notarisation] \label{2imp1}
If a block $b$ receives a stage-2 certificate or a stage-2 M-certificate, then $b$ receives a stage-1 notarisation.
\end{lemma}
\begin{proof}
In either case, at least one correct processor disseminates a stage-2 vote for $b$. Consider the first correct processor to do so. This processor must add $b$ to its local value $\mathtt{blocks}^*$ (see line~\ref{13}). Condition (i) in the definition of $\mathtt{blocks}^*$ requires that $p_i$ has received a stage-1 notarisation for $b$. Thus $b$ receives a stage-1 notarisation.
\end{proof}

\begin{lemma}[Consistency]
Carnot  1 satisfies Consistency.
\end{lemma}
\begin{proof}
Towards a contradiction, suppose some block $b_1$ with $b_1.\text{view}=v_1$ receives a stage-2 certificate, and that for some \emph{least} $v_2\geq v_1$ some block $b_2$ satisfies:
\begin{enumerate}
\item $b_2.\text{view}=v_2$;
\item $b_1$ is not an ancestor of $b_2$, and;
\item $b_2$ receives a stage-1 notarisation.
\end{enumerate}
From Lemmas~\ref{1certlem} and \ref{2imp1}, it follows that $v_2>v_1$. According to clause (i) from the definition of a votable fragment, correct processors will not disseminate a stage-1 vote for $b_2$ (line~\ref{vote1}) until adding the parent, $b_0$ say, to their local value $\mathtt{blocks}$. From Lemma  \ref{2imp1}, we conclude that  $b_0$ receives a stage-1 notarisation.  By our choice of $v_2$, it follows that $b_0.\text{view}<v_1$. This gives a contradiction, because by clause (ii) from the definition of a votable fragment, correct processors would not disseminate a stage-1 vote for $b_2$ without receiving an N-certificate for view $v_1$. By Lemma~\ref{2certlem}, such an N-certificate cannot exist. So, block $b_2$ cannot receive a stage-1 notarisation.
\end{proof}

\subsection{Liveness} \label{elivesec}

The proof of liveness is also straightforward. First we prove that correct processors progress through all views. Then we prove that correct leaders finalise new blocks after GST. 
Recall that $\delta$ is the (unknown) least upper bound on message delay after GST.  

\begin{lemma}[Timely view entry] \label{timev} 
If correct $p_i$ enters view $v$ at $t$, then all correct processors enter view $v$ by $t':=\max \{ t, \text{GST} \}+ \delta$. 
\end{lemma} 
\begin{proof} 
The proof is by induction on $v$. For $v=1$, the claim is immediate since all correct processors begin in view $1$. Suppose $v>1$ and that the claim holds for all previous views. Let $p_i$, $t$ and $t'$ be as in the statement of the lemma. Processor $p_i$ either enters view $v$ at $t$ upon adding a view $v-1$ block to $\mathtt{blocks}$, or upon receiving an N-certificate for view $v-1$. In the former case, $p_i$ has received a stage-2 M-certificate for some view $v-1$ block $b$ and for all non-genesis ancestors of $b$. In either case, $p_i$ disseminates all relevant certificates (lines~\ref{Ndis} and~\ref{2dis}), which are received by all correct processors by $t'$. By the induction hypothesis, all correct processors have entered view $v-1$ by $t'$, and so they all enter view $v$ by $t'$.
\end{proof}

\begin{lemma}[Progression through views] \label{progress} Every correct processor enters every view $v\in \mathbb{N}_{\geq 1}$.
\end{lemma}
\begin{proof}  Towards a contradiction, suppose that some correct processor $p_i$ enters view $v$, but never enters view $v+1$.  From Lemma \ref{timev}, it follows that:
\begin{itemize}
\item All correct processors enter view $v$;
\item No correct processor leaves view $v$.
\end{itemize}
Suppose first that at least $f+1$ correct processors disseminate stage-2 votes while in view $v$. No correct processor disseminates a stage-2 vote for any block $b$ before either receiving a stage-1 notarisation for $b$, or else adding $b$ to $\mathtt{blocks}$ (meaning that some correct processor has already disseminated a stage-2 vote for $b$). It follows from Lemma \ref{1certlem} that all stage-2 votes disseminated by correct processors while in view $v$ are for the same block $b$. The block $b$ therefore receives a stage-2 M-certificate, which is received by all correct processors. Since $b$ receives a stage-1 notarisation, and no correct processor would disseminate a stage-1 vote for $b$ without receiving and disseminating stage-2 M-certificates for all non-genesis ancestors of $b$, it follows that all correct processors enumerate $b$ into $\mathtt{blocks}$ and proceed to view $v+1$. This contradicts the claim that no correct processor leaves view $v$. 

So, suppose instead that at most $f$ correct processors disseminate stage-2 votes during view $v$. In this case, $\mathtt{TimeoutReady}$ will eventually be true for at least $n-2f$ correct processors. Those processors will therefore disseminate nullify$(v)$ messages. Since $n-2f \geq 2f+1$, all correct processors will receive an N-certificate for view $v$ and leave the view. This gives the required contradiction. 
\end{proof}

\begin{lemma}[Correct leaders finalise blocks] \label{L1} 
Suppose view $v$ is initial. If $p_i=\mathtt{lead}(v)$ is correct, and if the first correct processor to enter view $v$ does so at $t\geq \text{GST}$, then all correct processors receive a stage-2 certificate for some view $v$ block $b$ by $t+4\delta$, and also enter view $v+1$ by this time.  It also holds that no correct processor disseminates a nullify$(v)$ message in this case.  
\end{lemma}
\begin{proof} 
Suppose the conditions in the statement of the lemma hold. From Lemma~\ref{timev}, it follows that all correct processors (including $p_i$) enter view $v$ by $t+\delta$. Since $v$ is initial, $\mathtt{ProposeReady}$ is true when $p_i$ enters view $v$, so $p_i$ proposes a new block $b$ by $t+\delta$, and all correct processors receive their certified fragments of $b$ by $t+2\delta$. 

Let $b'$ be the parent of $b$ and set $v':=b'.\text{view}$. Since all ancestors of $b'$ (including $b'$) are in $p_i$'s local value $\mathtt{blocks}$ when $p_i$ proposes $b$, and since $p_i$ disseminates new stage-2 M-certificates (line~\ref{2dis}), every correct processor has all ancestors of $b'$ in their local value $\mathtt{blocks}$ by $t+2\delta$. Since $p_i$ has entered view $v$ by $t+\delta$, it must also have received N-certificates for all views in the open interval $(v',v)$ by this time. Since $p_i$ disseminates new N-certificates (line~\ref{Ndis}), all correct processors receive these by $t+2\delta$. 

All correct processors therefore have a votable fragment for $b$ (as defined in Section~\ref{formal}) by $t+2\delta$.   Since $\delta \leq \Delta$, this occurs before $\mathtt{TimeoutReady}$ becomes true via clause (c) in the definition of that predicate. Suppose first that no correct processor receives a stage-2 M-certificate for $b$ by $t+2\delta$. Then all correct processors disseminate stage-1 votes for $b$ (line~\ref{vote1}) by this time. In Section~\ref{formal}, we stipulated that, when $v$ is initial, the reconstruction parameter $k$ should be set to $k:=n-f-1$. Since all correct processors other than $p_i$ disseminate their certified fragments of $b$ upon voting (line~\ref{fdis}), all correct processors receive at least $n-f-1$ certified fragments of $b$, as well as a stage-1 notarisation for $b$, by $t+3\delta$. They therefore add $b$ to $\mathtt{blocks}^*$ by this time, and so disseminate stage-2 votes for $b$ by this time (either via line \ref{vote2} or \ref{vote2b}). 

If any correct processor receives a stage-2 M-certificate for $b$ by $t+2\delta$, it disseminates that certificate (line \ref{2dis}), and so all correct processors receive it by $t+3\delta$. All correct processors therefore  disseminate stage-2 votes for $b$ by this time (either via line \ref{vote2} or \ref{vote2b}). 

All correct processors therefore disseminate stage-2 votes for $b$ by $t+3\delta$. Since $\delta \leq \Delta$, this occurs before $\mathtt{TimeoutReady}$ becomes true via clause (d). All correct processors therefore receive a stage-2 certificate for $b$ by $t+4\delta$. They also receive a stage-2 M-certificate for $b$ and add $b$ to $\mathtt{blocks}$, and so enter view $v+1$ by this time. 
\end{proof}

For the following lemma, it is useful to introduce some terminology. Consider a view $v$ with correct leader $p_i$, and suppose $p_i$ proposes a view $v$ block $b$. Let $k$ be the reconstruction parameter for $b$, and let $I$ be the set of processors $p_j$ that disseminate a certified fragment of $b$ at $j$ upon receiving a votable fragment for view $v$. If $|I|\geq k$, then we say \emph{fragments are well disseminated in view $v$}. 

\begin{lemma}[Correct leaders finalise many blocks] \label{L2}
Suppose view $v$ is not initial, $p_i=\mathtt{lead}(v)$ is correct, and let $v_0$ be the first view in the same superview as $v$. Suppose the first correct processor to enter view $v_0$ does so at or after $\text{GST}$, and the first correct processor to enter view $v$ does so at $t$. Suppose further that fragments are well disseminated in all views in the interval $(v_0,v]$. Then all correct processors receive a stage-2 certificate for some view $v$ block by $t+3\delta$, and enter view $v+1$ by this time. Furthermore, no correct processor disseminates a nullify$(v)$ message.
\end{lemma}
\begin{proof} 
Suppose the conditions in the statement of the lemma hold. The proof is by induction on $v$. Suppose the claim holds for all non-initial views $v' < v$ in the same superview as $v$. Combined with Lemma~\ref{L1}, this means that no correct processor disseminates a nullify$(v')$ message for any $v' \in [v_0, v)$, and hence no correct processor disseminates a nullify$(v)$ message prior to entering view $v$. 

From Lemma~\ref{timev}, it follows that all correct processors enter view $v$ by $t+\delta$. By the induction hypothesis (or by Lemma~\ref{L1} if $v = v_0 + 1$), they must do so upon adding a view $v-1$ block to $\mathtt{blocks}$. In Section~\ref{formal}, we specified that $\mathtt{ProposeReady}$ must satisfy the following condition: there exists $t' < t$ such that $\mathtt{ProposeReady}=\text{true}$ as locally defined for $p_i$ at $t'$, and $p_i$ proposes a block $b$ for view $v$ at $t'$. It follows that all correct processors receive a votable fragment for view $v$ by $t+\delta$. Since $\delta \leq \Delta$, this occurs before $\mathtt{TimeoutReady}$ becomes true via clause (a). 

Suppose first that no correct processor receives a stage-2 M-certificate for $b$ by $t+\delta$. Then all correct processors  disseminate stage-1 votes for $b$ by $t+\delta$. 
Since fragments are well disseminated in view $v$, at least $k$ processors disseminate their certified fragments of $b$ by $t+\delta$, where $k$ is the reconstruction parameter for $b$. All correct processors therefore receive a stage-1 notarisation for $b$ and at least $k$ certified fragments of $b$ by $t+2\delta$, and so add $b$ to $\mathtt{blocks}^*$ by this time. Since $\delta \leq \Delta$, this occurs before $\mathtt{TimeoutReady}$ becomes true via clause (b). All correct processors therefore disseminate stage-2 votes for $b$ by $t+2\delta$ (via either line \ref{vote2} or \ref{vote2b}). 

If any correct processor receives a stage-2 M-certificate for $b$ by $t+\delta$, then all correct processors receive it by $t+2\delta$. Once again, this means all correct processors disseminate stage-2 votes for $b$ by $t+2\delta$ (via either line \ref{vote2} or \ref{vote2b}). 

In either case, all  correct processors receive a stage-2 certificate for $b$ by $t+3\delta$. They also receive a stage-2 M-certificate for $b$ and add $b$ to $\mathtt{blocks}$, and so enter view $v+1$ by this time, as claimed.
\end{proof}

\begin{lemma}[Liveness] \label{liveness}
Carnot 1 satisfies Liveness.
\end{lemma}
\begin{proof}
Suppose correct processor $p_i$ receives the transaction $\text{tr}$. By Lemma~\ref{progress}, all correct processors enter every view. Let $v$ be an initial view such that $p_i = \mathtt{lead}(v)$ and the first correct processor to enter view $v$ does so at or after $\text{GST}$, and after $p_i$ has received $\text{tr}$.  By Lemma~\ref{L1}, all correct processors receive a stage-2 certificate for some view $v$ block $b$ proposed by $p_i$.

From the definition of the ProposeBlock procedure,  $\text{tr}$ is included in the payload of $b$ or an ancestor of $b$. It remains to show that the payloads of $b$ and all ancestors of $b$ can be reconstructed from messages received by correct processors, i.e., that $\text{tr} \in \mathcal{F}(M_c(t))$ for some $t$.

 Each non-genesis ancestor $b'$ of $b$ (including $b$) is added to $ \mathtt{blocks}$ by every correct processor. By the definition of $\mathtt{blocks}$, $b'$ receives  a stage-2 M-certificate, which means at least $f+1$ processors disseminate stage-2 votes for $b'$ before seeing a stage-2 M-certificate for $b'$. If it has not already received a stage-2 M-certificate, no correct processor disseminates a stage-2 vote for $b'$ without first adding $b'$ to $\mathtt{blocks}^*$ (line~\ref{13}), which by condition (ii) in the definition of $\mathtt{blocks}^*$ requires successfully decoding the payload of $b'$. Thus at least one correct processor has received sufficient certified fragments to reconstruct the payload of $b'$, and these fragments are included in $M_c(t)$ for sufficiently large $t$.

By the definition of $\mathcal{F}$ in Section~\ref{formal}, it follows that $\text{tr} \in \mathcal{F}(M_c(t))$ for some $t$, as required.
\end{proof}

\subsection{Message and communication complexity}  \label{mc1} 
Within each view, each processor sends a constant-bounded number of messages to all others, meaning that message complexity is $O(n^2)$ per view. This is true for most standard protocols (PBFT, Tendermint, Simplex etc), with the notable exception of Hotstuff, which relays all messages via the leader, and so trades communication complexity for round complexity. 

\vspace{0.2cm} 
\noindent To analyse the communication complexity within each view, suppose hash values and signatures are of constant-bounded length. Then certified fragments for a block with payload of size $B$ are of size $O(Bd/n+\log n)$, where $d$ is the data expansion rate. So, the  leader sending certified fragments to all other processors induces communication complexity $O(Bd+n \log n)$. Then all other processors must disseminate their own fragments, inducing further communication complexity $O(Bdn+n^2 \log n)$. As analysed in Section \ref{pipes_analysis}, this round of fragment echoing not does really impact latency any more than the leader's initial sending of the block proposal, because the communication cost is shared between the processors. 

\vspace{0.2cm} 
\noindent In the formal specification  of Section \ref{2spec}, a vote for $b$ includes $b$. Replacing $ b$ with its hash value leads to votes of constant-bounded size, if hash values and signatures are of fixed length. Similarly, certificates are of constant-bounded length if hash lengths are fixed. So, the dissemination of these messages contributes communication complexity $O(n^2)$ per view. Whether or not votes include $b$ or the hash of $b$, the overall communication complexity per view is  $O(Bdn+n^2 \log n)$.

\section{Carnot 2: Analysis in the standard model} \label{anal3} 

In this section, we assume $n\geq 3f+1$, and consider the standard model of partial synchrony from Section \ref{setup}. 

\subsection{Consistency} \label{Ceconsec}
The proof of consistency follows a similar structure to that for  Carnot~1.

\begin{lemma}[One vote per view] \label{Csinglevote}
Correct processors disseminate at most one stage-1 vote and at most one stage-2 vote in each view.
\end{lemma}
\begin{proof}
For stage-1 votes, note that a correct processor $p_i$ only disseminates a stage-1 vote for a view $v$ block in line~\ref{2vote1}, and only if $1\mathtt{voted}(v) = \text{false}$. Upon doing so, $p_i$ sets $1\mathtt{voted}(v) := \text{true}$. For stage-2 votes, observe that $p_i$ only disseminates a stage-2 vote for a view $v$ block in line~\ref{2vote2}, and only if $2\mathtt{voted}(v) = \text{false}$. Upon doing so, $p_i$ sets $2\mathtt{voted}(v) := \text{true}$.
\end{proof}

\begin{lemma}[No stage-2 vote and nullify] \label{Cstage2ornull}
No correct processor both disseminates a stage-2 vote for a view $v$ block and sends a \emph{nullify}$(v)$ message.
\end{lemma}
\begin{proof}
A correct processor $p_i$ only disseminates a nullify$(v)$ message  (line~\ref{2time-out}) if  $2\mathtt{voted}(v)=\text{false}$, and sets $\mathtt{nullified}(v):=\text{true}$ upon doing so.  Similarly, a correct processor only disseminates a stage-2 vote for a view $v$ block (line \ref{2vote2}) if $\mathtt{nullified}(v)=\text{false}$, and sets $2\mathtt{voted}(v):=\text{true}$ upon doing so. 
\end{proof}

\begin{lemma}[Unique stage-1 notarisation per view] \label{C1certlem}
If block $b$ receives a stage-1 notarisation, then no block $b' \neq b$ with $b'.\text{view} = b.\text{view}$ receives a stage-1 notarisation.
\end{lemma}
\begin{proof}
This follows from Lemma~\ref{Csinglevote} by the standard quorum intersection argument. Suppose, towards a contradiction, that $b$ receives a stage-1 notarisation $Q$ and $b' \neq b$ with $b'.\text{view} = b.\text{view}$ receives a stage-1 notarisation $Q'$. Let $P$ and $P'$ be the sets of processors contributing to $Q$ and $Q'$, respectively. Then $|P \cap P'| \geq (n-f) + (n-f) - n = n - 2f \geq f + 1$. Thus $P \cap P'$ contains at least one correct processor, contradicting Lemma~\ref{Csinglevote}.
\end{proof}

\begin{lemma}[Stage-2 certificate precludes N-certificate] \label{C2certlem}
If block $b$ receives a stage-2 certificate, then view $v := b.\text{view}$ does not receive an N-certificate.
\end{lemma}
\begin{proof}
This follows from Lemma~\ref{Cstage2ornull} by the quorum intersection argument. Towards a contradiction, suppose $b$ receives a stage-2 certificate $Q$ and that  $v := b.\text{view}$ receives an N-certificate $Q'$.  Let $P$ and $P'$ be the sets of processors contributing to $Q$ and $Q'$, respectively. Then $|P \cap P'| \geq (n-f) + (n-f) - n =n-2f \geq f + 1$. Thus $P \cap P'$ contains at least one correct processor, contradicting Lemma~\ref{Cstage2ornull}.
\end{proof}


\begin{lemma}[Stage-2 certificates imply stage-1 notarisation] \label{C2imp1}
If a block $b$ receives a stage-2 certificate, then $b$ receives a stage-1 notarisation.
\end{lemma}
\begin{proof}
This follows since no correct processor disseminates a stage-2 vote for $b$ before receiving a stage-1 certificate for $b$ (required for addition to $\mathtt{blocks}$). 
\end{proof}

\begin{lemma}[Consistency]
Carnot  2 satisfies Consistency.
\end{lemma}
\begin{proof}
Let $M^*$ be the set of all messages received by at least one processor during the protocol execution. Towards a contradiction, suppose there exist inconsistent blocks $b_1$ and $b_3$ in $\mathtt{blocks}(M^*)$, both of which receive a stage-2 certificate. Without loss of generality, suppose $b_3.\text{view}\geq b_1.\text{view}$, and set $v_1:=b_1.\text{view}$. Then there exists some least $v_2\geq v_1$ such that, for  some block $b_2$: 
\begin{enumerate}
\item $b_2.\text{view}=v_2$;
\item $b_1$ is not an ancestor of $b_2$, and;
\item $b_2\in \mathtt{blocks}(M^*)$.
\end{enumerate}
From Lemmas~\ref{C1certlem} and \ref{C2imp1}, it follows that $v_2>v_1$. According to clause (iii) from the definition of $\mathtt{blocks}$, the parent of $b_2$, $b_0$ say, must be in  $\mathtt{blocks}(M^*)$.  By our choice of $v_2$, it follows that $b_0.\text{view}<v_1$. This gives a contradiction, because by clause (iv) from the definition of $\mathtt{blocks}$, $b_2$ cannot be added to $\mathtt{blocks}(M^*)$ unless $M^*$ contains an N-certificate for view $v_1$. By Lemma~\ref{C2certlem}, such an N-certificate cannot exist. 
\end{proof}

\subsection{Liveness} 

Let $\mathtt{blocks}_i$ be the variable $\mathtt{blocks}$ as locally defined for $p_i$, and let $\mathtt{blocks}_i(t)$ be $\mathtt{blocks}_i$ as defined at the end of timeslot $t$. We use similar notation w.r.t.\ the local value N-$\mathtt{certificates}$. 

\begin{lemma}[Agreement on $\mathtt{blocks}$] \label{2blockslem} 
If $p_i$ and $p_j$ are correct and $t$ is the least timeslot such that $b\in \mathtt{blocks}_i(t)$, then for $t':=\text{max} \{ t, \text{GST} \}+ 2\delta +s$, $b\in  \mathtt{blocks}_j(t')$. 
\end{lemma} 
\begin{proof} 
The proof is by induction on $b.\text{view}$. The result holds trivially for $b_{\text{gen}}$. So, suppose the conditions in the statement of the lemma hold, $v:=b.\text{view}>0$,  and that the claim holds for all previous views (and so, for the parent of $b$). 
 In this case,  
$p_i$ has received  a stage-1 certificate for $b$ by $t$, and so disseminates this by $t$ (line \ref{2Ndis}). It follows that all correct processors receive a stage-1 certificate for $b$ by $\text{max} \{ t, \text{GST} \}+ \delta$. If $b'$ is the parent of $b$ with $b'.\text{view}=v'$, then $p_i$ has also received N-certificates for all views in $(v',v)$, and has disseminated these by $t$  (line \ref{2Ndis}). At $t+s$, $p_i$ executes lines \ref{f1a}-\ref{f1b} of Algorithm \ref{alg4}  with respect to $b$. This ensures that all correct processors receive their certified recovery fragments of $b$ (or already know the block)   by  $\text{max} \{ t+s, \text{GST} \}+ \delta$. Lines \ref{sendowns}-\ref{f2b} ensure that correct processors other than $\mathtt{lead}(v)$ disseminate their recovery fragments to those who need it for reconstruction by the latter timeslot. All correct processors therefore receive a stage-1 certificate for $b$, N-certificates for all views in $(v',v)$, and fragments sufficient to decode the payload of $b$ by $\text{max} \{ t+s, \text{GST} \}+ 2\delta$. By the induction hypothesis, and since the parent of $b$ is in $\mathtt{blocks}_i(t)$, it follows that for $t':=\text{max} \{ t, \text{GST} \}+ 2\delta +s$ and for any correct processor $p_j$, $b\in  \mathtt{blocks}_j(t')$ as required. 
\end{proof} 

\begin{lemma}[Timely superview entry] \label{2timev} 
If correct $p_i$ enters superview $w$ at $t$, then all correct processors enter superview $w$ by $t':=\text{max} \{ t, \text{GST} \}+ 2\delta +s$. 
\end{lemma} 
\begin{proof} 
The proof is by induction on $w$. For $w=1$, the claim is immediate. Suppose $w>1$ and that the claim holds for all previous superviews.  Let $p_i$, $t$ and $t'$ be as in the statement of the lemma. By the induction hypothesis, it follows that all correct processors enter superview $w-1$ by $t'$. Processor $p_i$ enters superview $w$ at $t$ because N-$\mathtt{certificates}_i \cup \mathtt{blocks}_i$ contains either an N-certificate or a block for each view in superview $w-1$ (line \ref{newwclause}). The induction step then follows from Lemma \ref{2blockslem}, and since  $p_i$ disseminates new N-certificates upon receipt (line \ref{2Ndis}). 
\end{proof}

\begin{lemma}[Progression through superviews] \label{2progress} Correct processors enter every superview $w\in \mathbb{N}_{\geq 1}$.
\end{lemma}
\begin{proof}  Towards a contradiction, suppose that some correct processor $p_i$ enters superview $w$, but never enters superview $w+1$.  From Lemma \ref{2timev}, it follows that:
\begin{itemize}
\item All correct processors enter superview $w$;
\item No correct processor leaves superview $w$.
\end{itemize}
From Lemma  \ref{2blockslem}, and since correct processors disseminate new N-certificates upon receipt, it follows that for some view $v$ in superview $w$, no correct processor ever adds a view $v$ block to their local value $\mathtt{blocks}$, and no correct processor ever receives an N-certificate for view $v$. 
Then no correct processor disseminates a stage-2 vote for any view $v$ block, since they would add a view $v$ block to $\mathtt{blocks}$ upon doing so. $\mathtt{TimeoutReady}$ is therefore eventually  true for each  correct processor. All correct processors therefore disseminate nullify$(v)$ messages (line \ref{2time-out}), and so receive an N-certificate for view $v$. This gives the required contradiction. 
 \end{proof}

\begin{lemma}[Correct leaders finalise blocks] \label{2L1} 
Suppose view $v$ is the initial view of superview $w$. If $p_i=\mathtt{lead}(w)$ is correct, and if the first correct processor to enter superview $w$ does so at $t\geq \text{GST}$, then all correct processors add a view $v$ block $b$ to their local value $\mathtt{blocks}$ by  $t+4\delta+2s+s^*$, and receive a stage-2 certificate for $b$ by $t+5\delta+2s+s^*$.  It also holds that $\mathtt{TimeoutReady}(v)$ is never true for any correct processor in this case, so that no correct processor disseminates a nullify$(v)$ message.  
\end{lemma}
\begin{proof} 
Suppose the conditions in the statement of the lemma hold. From Lemma~\ref{2timev}, it follows that all correct processors (including $p_i$) enter superview $w$ by $t+2\delta+s$. $\mathtt{ProposeReady}$ is true when $p_i$ enters superview $w$, and  $\mathtt{b}$ as locally defined for $p_i$ is in $\mathtt{blocks}_i$ at this time, and so is $\mathtt{blocks}$ for all correct processors by $t+4\delta+2s$ (by Lemma \ref{2blockslem}). We also stipulated that, when $v$ is initial, the reconstruction parameter $k$ should be set to $k:=n-f-1$. Let $v':=\mathtt{b}.\text{view}$. Since $p_i$ has entered superview $w$, it has received N-certificates for all views in $(v',v)$, and all correct processors receive these by $t+3\delta+s$. As stated in Section \ref{2spec}, our formal assumption is that $p_i$ sends a proposal $b$ by $t+2\delta+s+s^*$.\footnote{In the standard model of partial synchrony, we could set $s^*=0$, but we also wish timing for the superview to make sense in a context where bandwidth is limited.} So  all correct processors receive their certified fragments of $b$ by $t+3\delta+s+s^*$ and disseminate stage-1 votes for $b$  by this time (before their local value $\mathtt{TimeoutReady}(v)$ becomes true through clause (a)). All correct processors other than $p_i$ also disseminate their corresponding fragments of $b$ by this time. 

It follows that, by  $t+4\delta+2s+s^*$, all correct processors add the parent of $b$ to $\mathtt{blocks}$,  receive a stage 1 certificate for $b$ along with certified fragments of $b$ sufficient to decode $b$, and also receive N-certificates for all views in $(v',v)$. All correct processors therefore add $b$ to $\mathtt{blocks}$ and disseminate stage-2 votes for $b$ by this time (before their local value $\mathtt{TimeoutReady}(v)$ becomes true through clause (b)).

All correct processors therefore receive a stage-2 certificate for $b$ by $t+5\delta+2s+s^*$ before their local value $\mathtt{TimeoutReady}(v)$ becomes true through clause (c), meaning that no correct processor disseminates a nullify$(v)$ message. 
\end{proof}

Recall the following terminology from Section \ref{formal}, which we now modify to fit with Carnot ~2. Consider a view $v$ in superview $w$ with correct leader $p_i$, and suppose $p_i$ proposes a view $v$ block $b$. Let $k$ be the reconstruction parameter for $b$, and let $I$ be the set of processors $p_j$ that disseminate a certified fragment of $b$ at $j$ upon receiving such a fragment from $p_i$. If $|I|\geq k$, then we say \emph{fragments are well disseminated in view $v$}. 

\begin{lemma}[Correct leaders finalise many blocks] \label{2L2} Let $v_1,\dots,v_x$ be the views of superview $w$ (in order from least to greatest). Suppose  $p_i=\mathtt{lead}(w)$ is correct, and that the first correct processor to enter superview $w$ does so at $t\geq \text{GST}$. For some $j\in (1,x]$, suppose further that fragments are well disseminated in all views in the interval $(v_1,v_j]$. Then all correct processors  add a view $v_j$ block, $b_j$ say,  to their local value $\mathtt{blocks}$ by $t+4\delta+2s+js^*$, and receive a stage-2 certificate for $b_j$ by $t+5\delta+2s+js^*$. Furthermore, $\mathtt{TimeoutReady}(v_j)$ is never true for any correct processor in this case, so that no correct processor disseminates a nullify$(v_j)$ message.  
\end{lemma}
\begin{proof} 
Suppose the conditions in the statement of the lemma hold. The proof is by induction on $j>1$. Suppose the claim holds for all  $j' \in (1,j)$. Combined with Lemma~\ref{2L1}, this means that no correct processor disseminates a nullify$(v_{j'})$ message for any $j'<j$, and that no correct processor disseminates a  nullify$(v_j)$ message because $\mathtt{TimeoutReady}(v_{j'})$ is true for $j'<j$. All correct processors add blocks $b_1,\dots,b_{j-1}$ to $\mathtt{blocks}$ by $t+4\delta+2s+(j-1)s^*$ (where each $b_{j'}$ is a block for view $v_{j'}$). 

As stated in Section \ref{2spec}, our formal assumption is that $p_i$ sends the proposal $b_j$ by $t+2\delta+s+js^*$. So  all correct processors receive their certified fragments of $b_j$ by $t+3\delta+s+js^*$ and disseminate stage-1 votes for $b_j$  by this time, before their local value $\mathtt{TimeoutReady}(v_j)$ becomes true through clause (a). All correct processors other than $p_i$ also disseminate their corresponding fragments of $b_j$ by this time. 

Since  fragments are well disseminated in view $v_j$, it follows that, by  $t+4\delta+2s+js^*$, all correct processors receive a stage 1 certificate for $b_j$ along with certified fragments of $b_j$ sufficient to decode it.  Since the parent of $b_j$ is already in $\mathtt{blocks}$, correct processors therefore add $b_j$ to $\mathtt{blocks}$ and disseminate stage-2 votes for $b_j$ by this time, before their local value $\mathtt{TimeoutReady}(v_j)$ becomes true through clause (b).

All correct processors therefore receive a stage-2 certificate for $b_j$ by $t+5\delta+2s+js^*$ before their local value $\mathtt{TimeoutReady}(v_j)$ becomes true through clause (c), meaning that no correct processor disseminates a nullify$(v_j)$ message. 
\end{proof}

\begin{lemma}[Liveness] \label{2liveness}
Carnot 2 satisfies Liveness.
\end{lemma}
\begin{proof}
Suppose correct processor $p_i$ receives the transaction $\text{tr}$. By Lemma~\ref{2progress}, all correct processors enter every superview. Let $v$ be the initial view of some superview $w$ with $p_i = \mathtt{lead}(w)$,  such that the first correct processor to enter superview $w$ does so at or after $\text{GST}$, and after $p_i$ has received $\text{tr}$.  By Lemma~\ref{2L1}, all correct processors add some view $v$ block $b$ to $\mathtt{blocks}$, and receive  a stage-2 certificate for $b$. 
From the definition of the ProposeBlock procedure,  $\text{tr}$ is included in the payload of $b$ or an ancestor of $b$, meaning that $\text{tr}$ is finalised by all correct processors. 
\end{proof}

\subsection{Communication complexity} 

Message and communication complexity are the same as for Carnot~1, except that each processor may be required to send up to two recovery fragments to each other processor if others are not able to reconstruct the block. For blocks with payload of size $B$ and a data expansion rate $d$, this still induces communication cost $O(Bdn+n^2 \log n)$ per view, meaning that the overall communication complexity per view remains $O(Bdn+n^2 \log n)$. A much more detailed analysis (in the pipes model) appears in Section \ref{anal2}.

\section{Pipes analysis for Dispersed Simplex} \label{pipesDSanal} 
\noindent We carry out the analysis in the pipes model for the
formally specified multi-view-leader version of DispersedSimplex with data
expansion rate $1.5$. As above, we restrict to the good case. In 
this version of the protocol, the leader enters a view at time $t$ 
upon receiving a support certificate for the previous view. Support 
certificates contain a hash and a signature, so we model them as 
being of length $3\lambda$.  Block proposal messages contain a 
fragment together with a hash and a validation path, so we model 
them as being of length $(1+\log n)\lambda + 1.5B/n$, where $B$ is 
the payload size and $1.5$ is the data expansion rate. Support 
shares additionally include a signature share, so we model these as 
being of length $(3+\log n)\lambda + 1.5B/n$. As above, we assume 
that processors begin forwarding fragments immediately upon receipt. 

\vspace{0.2cm} 
\noindent Suppose the leader enters a view at $t$.  It first 
disseminates the support certificate for the previous view, taking 
time $3n\lambda/S$. It then sends out block proposal messages, which 
other processors receive by:
\[ t + \frac{3n\lambda}{S} + \frac{1.5B}{S} + \frac{(1+\log n)n\lambda}{S} + \delta. \] 
By our forwarding assumption, other processors begin forwarding 
upon receipt, but each support share includes an additional 
signature share relative to the block proposal message. The 
additional $2\lambda$ per support share takes a further 
$2n\lambda/S$ to be removed from upload buffers and $\delta$ to be 
received by others. The view therefore takes time 
\[ T:= \frac{1.5B}{S} + \frac{(6+\log n)n\lambda}{S} + 2\delta. \] 
At equilibrium, $B=DT$. Substituting $B = DT$:
\[ T = \frac{1.5DT}{S} + \frac{(6+\log n)n\lambda}{S} + 2\delta, \]
which rearranges to 
\[ T\left( 1 - \frac{1.5D}{S} \right) = \frac{(6+\log n)n\lambda}{S} + 2\delta. \]
\noindent Equivalently, we can write 
\[ T = \left( \frac{(6+\log n)n\lambda}{S} + 2\delta \right) \cdot \frac{1}{1 - 1.5D/S}. \] 
Since $T \to \infty$ as $D \to 2S/3$, the latency bottleneck for 
this protocol is $D = 2S/3$, corresponding to the data expansion 
rate of $1.5$.

\vspace{0.2cm} 
\noindent To obtain total latency, we must account for: (i) the 
time (at most) $T$ that each transaction waits to be included in a 
block; (ii) the time $T$ for the view containing the block to play 
out; and (iii) the time required for the block to be committed. 
Once $B_v$ has been added to the complete block tree, parties 
broadcast a commit share for $v$ --- a message containing a hash and 
a signature share, of size approximately $3\lambda$ --- and the 
block is committed once a commit certificate is formed. This takes 
a further $3n\lambda/S + \delta$. Total latency is therefore: 
\[ 2T + \frac{3n\lambda}{S} + \delta = \left( \frac{(6+\log n)2n\lambda}{S} + 4\delta \right) \cdot \frac{1}{1 - 1.5D/S} + \frac{3n\lambda}{S} + \delta. \]
For reasonable values of $n$ and $S$, the $(6+\log n)n\lambda/S$ and 
$3n\lambda/S$ terms are small. For incoming transaction rates well 
below the bottleneck $2S/3$, the factor $1/(1-1.5D/S)$ is close to 
$1$, and so total latency is approximately $5\delta$. As $D$ 
approaches $2S/3$, the factor $1/(1-1.5D/S)$ tends to infinity, and 
so does latency.

\end{document}